\documentclass{article}
\usepackage{fullpage}
\usepackage{times}
\usepackage{color}
\usepackage{amsmath,amssymb,amsthm, amsfonts, enumitem,centernot,mathtools,extarrows,mathdots,appendix}
\usepackage{comment,MnSymbol}
\usepackage[ruled,vlined]{algorithm2e}
\usepackage{physics}
\usepackage{tikz-cd}

\newtheorem{theorem}{Theorem}
\newtheorem{lemma}{Lemma}

\newtheorem{proposition}{Proposition}

\begin{document}
\title{Local Unitary Equivalence of Tripartite Quantum States In Terms of Trace Identities}
\author{Isaac Dobes and Naihuan Jing}
\maketitle

\begin{abstract}
    In this paper, we generalize the major results from Jing-Yang-Zhao's paper "Local Unitary Equivalence of Quantum States and Simultaneous Orthogonal Equivalence," which established a correspondence between local unitary (LU) equivalence and simultaneous orthogonal (SO) equivalence of bipartite quantum states, and then used this correspondence to show that the problem of determining LU equivalence of bipartite states reduces to checking trace identities. In particular, we extend both results to tripartite quantum states. We are able to do this by utilizing a hypermatrix algebra framework and by applying a powerful generalization of Specht's criterion proved in Futorny-Horn-Sergeichuk's paper "Specht's Criterion for Systems of Linear Mappings." With our established hypermatrix algebra framework and the aforementioned generalization of Specht's criterion, it is apparent that our results can be extended to arbitrary multipartite quantum states, however there are some practical limitations which are explored and discussed towards the end. 
\end{abstract}

\section{Introduction}
Quantum entanglement is a peculiar phenomenon that has intrigued yet also eluded researchers for nearly a century \cite{einstein1935can}. In the last few decades, however, serious interest and efforts in understanding entanglement has grown rapidly due to its applications to quantum computing \cite{nielsen2001quantum,bengtsson2017geometry,yu2021advancements}. Since entanglement pertains to nonlocal properties of quantum states \cite{walter2016multipartite}, a choice of local basis should have no effect on the degree to which a quantum state is entangled \cite{monras2011entanglement}. Therefore, one important task in the study of entanglement is the classification of quantum states up to local unitary (LU) equivalence. 

Two density matrices $\rho$ and $\widehat{\rho}$ representing states in the same quantum system $\mathbb{C}^{d_1}\otimes...\otimes \mathbb{C}^{d_N}$ are \textbf{locally unitary (LU) equivalent} if there exist $U_i\in SU(d_i)$ ($1\leq i\leq N$) such that \cite{li2014local}
\[\widehat{\rho} = (U_1\otimes...\otimes U_N)\rho(U_1\otimes...\otimes U_N)^{\dagger}.\]
In general, it is difficult to determine LU equivalence between two density matrices. One of the first major accomplishments in characterizing LU equivalence was due to Yuriy Makhlin in \cite{makhlin2002nonlocal}, where he provided necessary and sufficient conditions for the LU equivalence of $2$-qubit density matrices in terms of their Fano form. The next major development came from Barbara Kraus in \cite{kraus2010local2}, where she characterized LU equivalence of pure $n$-qubit states in terms of their "standard form." Deriving a state's "standard form" involved diagonalizing the corresponding reduced states, and provided an efficient way to check for LU equivalence of pure $n$-qubit states; however, there are degenerate cases where Kraus's methods breaks down. Inspired by Kraus's work, Jun-Li Li, Cong-Feng Qiao, and others represent arbitrary multipartite quantum states and characterize their LU equivalence in terms of their higher order singular value decomposition \cite{liu2012local}. Li and Qiao later extend their results to mixed multipartite states of arbitrary dimension \cite{li2013classification}, however as with Kraus's method there is a degenerate cases in which their method is very computationally taxing. An alternative approach to characterizing LU equivalence in terms of trace identities is later given by Jing et al. in \cite{jing2016local}. In their paper, they establish a correspondence between simultaneous orthogonal equivalence and quasi-LU equivalence (a property that follows from LU equivalence and in fact implies LU equivalence in certain cases for $2$-qubit density matrices), and then apply a generalization of Specht's criterion \cite{jing2015unitary} to show that quasi-LU equivalence of bipartite quantum states reduces to checking trace identities. This provides a nice and efficient way to check for LU equivalence, however there is a minor error in establishing the LU equivalence in the case of $2$-qubit density matrices, and furthermore their methods only apply to bipartite quantum states. 

In this paper, we present a modified version of the proof in \cite{jing2016local}, utilizing a hypermatrix algebra framework to establish a correspondence between between simultaneous orthogonal equivalence and quasi-LU equivalence of bipartite quantum states, and also LU equivalence in the case of $2$-qubit density matrices. Since our proof is written in the language of hypermatrices, this allows us to generalize the correspondence between simultaneous orthogonal equivalence and quasi-LU equivalence of tripartite quantum states, and again LU equivalence in the case of $3$-qubit density matrices. We then apply a more powerful generalization of Specht's criterion \cite{futorny2017specht} to reduce the problem of determining quasi-LU equivalence of tripartite quantum states (or LU equivalence in the case of $3$-qubit density matrices) to trace identities. Indeed, it is apparent that our proofs can, with relative ease, be further generalized to arbitrary dimensional multipartite states, however, there are computational limits to the practicality of our results, hence we stop at the tripartite case. 

\section{Preliminaries}
\subsection{Hypermatrix Algebra}
Most of the concepts and terminology in this subsection can be found \cite{lim2013tensors}. Let $F$ be a field and suppose $V$ and $W$ are finite-dimensional vector spaces of dimensions $m$ and $n$, respectively. Just as a linear transformation $T:V\rightarrow W$ can be represented as a matrix $A\in F^{m\times n}$, similarly a tensor $T\in V_1\otimes...\otimes V_d$, where $\dim(V_i)=n_i$ for each $1\leq i\leq d$, can be represented as a \textbf{hypermatrix} $A\in F^{n_1\times...\times n_d}$. The positive integer $d$ is called the \textbf{order} of a hypermatrix. 

Suppose $A\in F^{n_1\times...\times n_d}$ and $B\in F^{k_1\times...\times k_e}$ are hypermatrices of order $d$ and $e$. The \textbf{outer product} of $A$ and $B$, denoted $A\circ B$, is the hypermatrix of order $d+e$ whose $(i_1,...,i_d,j_1,...,j_e)$ coordinate is given by $a_{i_1,...,i_d}b_{j_1,...,j_e}$. For a hypermatrix $A\in F^{n_1\times n_2\times...\times n_d}$ and matrices $X_1\in F^{\bullet\times n_1},...,X_d\in F^{\bullet\times n_d}$, the \textbf{multilinear matrix multiplication} of $(X_1,...,X_d)$ with $A$ is defined to be the order $d$ hypermatrix
$A' := (X_1,...,X_d)*A\in F^{n_1\times...\times n_d}$, 
such that
\[A'_{i_1i_2...i_d} = \sum\limits_{j_1,j_2,...,j_d=1}^{n_1,n_2,...,n_d}(X_1)_{i_1j_1}...(X_d)_{i_dj_d}A_{j_1j_2...j_d}.\]
Multilinear matrix multiplication is linear in terms of the matrices in both parts; that is, if $\alpha,\beta\in F$, $X_1,Y_1\in F^{m_1\times n_1}$;...; $X_d,Y_d\in F^{m_d\times n_d}$; and $A,B\in F^{n_1\times n_2\times...\times n_d}$; then 
\[(X_1,...,X_d)*(\alpha A + \beta B) = \alpha(X_1,...,X_d)*A + \beta(Y_1,...,Y_d)*B\]
and 
\[[\alpha(X_1,...,X_d)+\beta(Y_1...,Y_d)]*A = \alpha(X_1,...,X_d)*A + \beta(Y_1,...,Y_d)*B.\]
Next, $k:[d]\rightarrow \mathbb{N}$. The outer product interacts with multilinear matrix multiplication in the following way
\begin{eqnarray}\label{outerproduct - multilinmatmult - relation}
    \begin{split}
        &(X_{1_1},...,X_{1_{k(1)}},...,X_{d_1},...,X_{d_{k(d)}})*(A_1\circ...\circ A_d) \\
        &\qquad = (X_{1_1},...,X_{1_{k(1)}})*A_1\circ...\circ (X_{d_1},...,X_{d_{k(d)}})*A_d
    \end{split} 
\end{eqnarray}
where $X_{i_j}\in F^{m^{(i)}_j\times n^{(i)}_j}$ and $A_i\in F^{n^{(i)}_1\times....\times n^{(i)}_{k(i)}}$. Consequently, for an arbitrary hypermatrix $A\in F^{n_1\times...\times n_d}$, which can be expressed as
\[A = \sum\limits_{k=1}^r\alpha_k(v_1^{(k)}\circ...\circ v_d^{(k)})\]
for some $r\in \mathbb{N}$ and collection of vectors $v_i^{(k)}\in F^{n_i}$ for each $i$ and $k$, then
\begin{equation}\label{outerproduct - multilinmatmult - special case}
    (X_1,...,X_d)*A = \sum\limits_{k=1}^r\alpha_k(X_1v_1^{(k)})\circ...\circ (X_dv_d^{(k)})
\end{equation}
for any $X_1\in F^{\bullet\times n_1}$,..., $X_d\in F^{\bullet\times n_d}$.

In some cases, it is more convenient to view hypermatrices as matrices. The \textbf{k-mode unfolding} \cite{kolda2009tensor} of a hypermatrix $A\in F^{n_1\times n_2\times...\times n_d}$ is the $n_k\times (n_{k+1}...n_dn_1...n_{k-1})$ matrix, denoted $A_{(k)}$, whose $(i_k,j)$ entry is given by $(i_1,...,i_d)$-entry of $A$, with
\[j=1+\sum\limits_{\substack{l=1 \\ l\neq k}}^d\left[(i_l-1)\prod\limits_{\substack{m=1 \\ m\neq k}}^{l-1}n_m\right],\]
or in the case where the index starts at $0$, 
\[j = \sum\limits_{\substack{l=1 \\ l\neq k}}^d\left[i_l\prod\limits_{\substack{m=1 \\ m\neq k}}^{l-1}n_m\right].\]
For example, if $A = [A_{i_1i_2i_3}]\in F^{2\times 2\times 2}$, then it can be verified that $A_{(1)}$ is the $2\times 4$ matrix given by
\[A_{(1)} = \left[\begin{array}{cccc}
    A_{111} & A_{121} & A_{112} & A_{122} \\
    A_{211} & A_{221} & A_{212} & A_{222}
\end{array}\right].\]
We call the process of unfolding a hypermatrix $A$ to a matrix $A_{(k)}$ the \textbf{matricization} of $A$. 

\subsection{Cayley's Second Hyperdeterminant}
Lastly, we review the \textbf{Cayley's second hyperdeterminant}, also known as the \textbf{geometric hyperdeterminant} \cite{lim2013tensors}. Let $A \in F^{n_1\times...\times n_d}$ be an order $d$ hypermatrix. The multilinear functional associated with $A$ if defined as 
\begin{align*}
    f_A:F^{n_1\times...\times n_d}&\longrightarrow F \\
    (\mathbf{x}_1,...,\mathbf{x}_d) &\mapsto \sum\limits_{j_1,...,j_d=1}^{n_1,...,n_d}a_{j_1...j_d}\mathbf{x}_1^{j_1}...\mathbf{x}_d^{j_d}.
\end{align*}
The geometric determinant of $A$, denoted as $\mathrm{Det}(A)$, is a homogeneous polynomial in the entries of $A$ such that $\mathrm{Det}(A) = 0$ if and only if the system of multilinear equations
\[\nabla f_A = 0\]
has a nontrivial solution. The geometric hyperdeterminant exists for any hypermatrix $A\in F^{n_1\times...\times n_d}$ if 
\[n_k-1 \leq \sum\limits_{j\neq k}(n_j-1)\]
for each $k=1,...,d$; consequently, it exists for all cuboid hypermatrices $A\in F^{\overbrace{n\times...\times n}^{d\text{ times}}}$ of order $d\geq 2$. When it exists, $\mathrm{Det}(A)$ satisfies an important property stated below:
\begin{proposition}\cite{lim2013tensors}
	Let $A\in F^{n_1\times...\times n_d}$ and suppose $\mathrm{Det}(A)$ exists. Then for any $X_1\in GL(n_1)$,..., $X_d\in GL(n_d)$, we have that
    \[\mathrm{Det}((X_1,...,X_d)*A) = \det(X_1)^{m/n_1}...\det(X_d)^{m/n_d}\mathrm{Det}(A),\]
    where $m$ is the degree of $\mathrm{Det}(A)$. Consequently, $\mathrm{Det}$ is invariant under the multilinear matrix multiplication of $(X_1,...,X_d)$ with $X_i\in SL(n_i)$ (and in particular, $X_i\in SU(n_i)$) for each $i$. 
\end{proposition}

\section{Quasi-LU Equivalence}
Let $\rho$ be the density matrix of a (mixed or pure) multipartite state of the quantum system $\mathbb{C}^{d_1}\otimes...\otimes\mathbb{C}^{d_n}$ and let $\{\lambda_i^k:1\leq i\leq d_k^2-1; k=1,...,n\}$ denote the set of $d_k\times d_k$ generalized Gell-Mann (GGM) matrices \cite{bertlmann2008bloch}. Then recall we can express $\rho$ as \cite{de2011multipartite}:
\begin{align*}
	&\quad \frac{1}{d_1\cdot...\cdot d_n}\bigg(\bigotimes_{k=1}^nI_{d_k}
    +\sum\limits_{j_1=1}^n\sum\limits_{\alpha_1=1}^{d_{j_1}^2-1}T_{j_1}^{\alpha_1}\lambda_{\alpha_1}^{(j_1)}+\sum\limits_{1\leq j_1<j_2\leq n}\sum\limits_{\alpha_1=1}^{d_{j_1}^2-1}\sum\limits_{\alpha_2=1}^{d_{j_2}^2-1}T_{j_1j_2}^{\alpha_1\alpha_2}\lambda_{\alpha_1}^{(j_1)}\lambda_{\alpha_2}^{(j_2)}+.... \\
	&\qquad ....+\sum\limits_{1\leq j_1<...<j_m\leq n}\sum\limits_{\alpha_1=1}^{d_{j_1}^2-1}\sum\limits_{\alpha_2=1}^{d_{j_2}^2-1}...\sum\limits_{\alpha_m=1}^{d_{j_m}^2-1}T_{j_1j_2...j_m}^{\alpha_1\alpha_2...\alpha_m}\lambda_{\alpha_1}^{(1)}\lambda_{\alpha_2}^{(2)}...\lambda_{\alpha_m}^{(j_m)}+....\\
	&\qquad ....+\sum\limits_{\alpha_1=1}^{d_1^2-1}...\sum\limits_{\alpha_n=1}^{d_n^2-1}T_{12...n}^{\alpha_1\alpha_2...\alpha_n}\lambda_{\alpha_1}^{(1)}\lambda_{\alpha_2}^{(2)}...\lambda_{\alpha_n}^{(n)}\bigg)
\end{align*}
where 
$\lambda_{\alpha_i}^{(j_k)} = I_{d_1}\otimes...\otimes\lambda_{\alpha_i}^{j_k}\otimes...\otimes I_{d_n}$ (i.e. $\lambda_{\alpha_i}^{j_k}$ is the $j_k^{th}$ factor in the tensor product) and 
\[T_{j_1j_2...j_m}^{\alpha_1\alpha_2...\alpha_m}=\text{tr}(\rho\sigma_{\alpha_1}^{(j_1)}\lambda_{\alpha_2}^{(j_2)}...\lambda_{\alpha_m}^{(j_m)}),\qquad m\leq n,\]
are all real coefficients. In particular,
\[T_{j_1} = \left[T_{j_1}^1,...,T_{j_1}^{d_{j_1}^2-1}\right]^t\] are real-valued vectors, 
\[T_{j_1j_2} = \bigg[T_{j_1j_2}^{\alpha_1\alpha_2}\bigg]_{(d_{j_1}^2-1)\times (d_{j_2}^2-1)}\]
are real-valued matrices, and in general, 
\[T_{j_1j_2...j_m} = \bigg[T_{j_1j_2...j_m}^{\alpha_1\alpha_2...\alpha_m}\bigg]_{(d_{j_1}^2-1)\times...\times (d_{j_m}^2-1)}\]
are real-valued hypermatrices of order $m$. These vectors, matrices, and hypermatrices uniquely define a multipartite quantum state, and so we call the set 
\[\{T_{j_1...j_m}:1\leq j_1<...<j_m\leq n; 1\leq m\leq n\}\] 
the \textbf{hypermatrix representation} of $\rho$. Also, for convenience, from now on we will denote $d_k^2-1$ as $\delta_k$. 

Let $\rho$ and $\widehat{\rho}$ be density matrices of multipartite states of the quantum system $\mathbb{C}^{d_1}\otimes...\otimes\mathbb{C}^{d_n}$. Recall that $\rho$ and $\widehat{\rho}$ are \textbf{locally unitary (LU) equivalent} if there exist $U_i\in SU(d_i)$, $1\leq i\leq n$ such that 
\begin{equation}\label{LU equivalence defn}
    \widehat{\rho} = (U_1\otimes...\otimes U_n)\rho(U_1\otimes...\otimes U_n)^{\dagger}.
\end{equation}
Now, given a $d_k\times d_k$ GGM matrix $\lambda_i^k$, we have that 
\begin{equation}\label{unitary-to-orthogonal}
	U_k\lambda_i^kU_k^{\dagger} = \sum\limits_{j=1}^{\delta_k}X_{ij}\lambda_j^k
\end{equation}
for some matrix $X = [X_{ij}]_{\delta_k\times \delta_k}$. Since the GGM matrices are Hermitian, 
\[U_k\lambda_i^kU_k^{\dagger} = U_k(\lambda_i^k)^{\dagger}U_k^{\dagger} = (U_k\lambda_i^kU_k^{\dagger})^{\dagger},\]
from which it follows that the coefficients $x_{ij}$ are real numbers. Moreover, 
\[\Tr((U_k\lambda_i^kU_k^{\dagger})(U_k\lambda_j^kU_k^{\dagger})) = \Tr(\lambda_i^k\lambda_j^kU^{\dagger}U) = \Tr(\lambda_i^k\lambda_j^k) = \delta_{ij},\]
where $\delta_{ij}$ denotes the Kronecker delta and the last equality follows from the fact that the GGM matrices are orthogonal with respect to the Hilbert-Schmidt inner product. Thus, it follows that the matrix $X$ is orthogonal; denote it as $\widetilde{O}_k$. Then from the linearity of the Kronecker product and the identity below:
\[(A_1\otimes A_2\otimes...\otimes A_n)(B_1\otimes B_2\otimes...\otimes B_n) = (A_1B_1)\otimes (A_2B_2)\otimes...\otimes (A_nB_n),\]
for each $m$ we have that 
\begin{align*}
	&(U_1\otimes...\otimes U_n)\left(\sum\limits_{1\leq j_1<...<j_m\leq n}\sum\limits_{\alpha_1,...,\alpha_m=1}^{\delta_{j_1},...,\delta_{j_m}}T_{j_1...j_m}^{\alpha_1...\alpha_m}\lambda_{\alpha_1}^{(1)}...\lambda_{\alpha_m}^{(j_m)}\right)(U_1\otimes...\otimes U_n)^{\dagger} \\
	&= \sum\limits_{1\leq j_1<...<j_m\leq n}\sum\limits_{\alpha_1,...,\alpha_m=1}^{\delta_{j_1},...,\delta_{j_m}}T_{j_1...j_m}^{\alpha_1...\alpha_m}(U_1\otimes...\otimes U_n)\lambda_{\alpha_1}^{(1)}...\lambda_{\alpha_m}^{(j_m)}(U_1\otimes...\otimes U_n)^{\dagger} \\
	&= \sum\limits_{1\leq j_1<...<j_m\leq n}\sum\limits_{\alpha_1,...,\alpha_m=1}^{\delta_{j_1},...,\delta_{j_m}}T_{j_1...j_m}^{\alpha_1...\alpha_m}(U_{j_1}\lambda_{\alpha_1}U_{j_1}^{\dagger})^{(j_1)}...(U_{j_m}\lambda_{\alpha_m}U_{j_m}^{\dagger})^{(j_m)} \\
	&= \sum\limits_{1\leq j_1<...<j_m\leq n}\sum\limits_{\alpha_1,...,\alpha_m=1}^{\delta_{j_1},...,\delta_{j_m}}T_{j_1...j_m}^{\alpha_1...\alpha_m}\left(\sum\limits_{k_1=1}^{\delta_{j_1}}(\widetilde{O}_{j_1})_{\alpha_1k_1}\lambda_{k_1}\right)^{(j_1)}...\left(\sum\limits_{k_m=1}^{\delta_{j_m}}(\widetilde{O}_{j_m})_{\alpha_mk_m}\lambda_{k_m}\right)^{(j_m)} \\
    &\qquad \text{by equation }\eqref{unitary-to-orthogonal} \\
	&= \sum\limits_{1\leq j_1<...<j_\leq n}\left(\sum\limits_{k_1,...,k_m=1}^{\delta_{j_1},...,\delta_{j_m}}\right)\left(\sum\limits_{\alpha_1,...,\alpha_m=1}^{\delta_{j_1},...,\delta_{j_m}}(\widetilde{O}_{j_1}^t)_{k_1\alpha_1}...(\widetilde{O}_{j_m}^t)_{k_m\alpha_m}T_{j_1...j_m}^{\alpha_1...\alpha_m}\right)\lambda_{k_1}^{(1)}...\lambda_{k_m}^{(j_m)} \\
	&= \sum\limits_{1\leq j_1<...<j_m\leq n}\left(\sum\limits_{k_1,...,k_m=1}^{\delta_{j_1},...,\delta_{j_m}}\right)\left((O_{j_1},...,O_{j_m})*T_{j_1...j_m}\right)^{k_1...k_m}\lambda_{k_1}^{(j_1)}...\lambda_{k_m}^{(j_m)} \\
    &\qquad \text{setting }O_{j_i} := \widetilde{O}_{j_i}^t.
\end{align*}
Thus, we say that $\rho$ and $\widehat{\rho}$ are \textbf{quasi-LU equivalent} if there exists $O_{j_1}\in O(\delta_{j_1}),...,O_{j_m}\in O(\delta_{j_m})$ such that 
\begin{equation}\label{quasi-LU equivalence defn}
    \widehat{T}_{j_1...j_m} = (O_{j_1},...,O_{j_m})*T_{j_1...j_m}\qquad (1\leq j_1<...<j_m\leq n)
\end{equation}
for each $1\leq m\leq n$. As shown above, LU equivalence implies quasi-LU equivalence. In the case of density matrices of $n$-qubits, \eqref{unitary-to-orthogonal} defines a surjective map \[SU(2)\rightarrow SO(3),\] and so in this case if $O_1,...,O_n$ are special orthogonal, then the converse holds as well (i.e. quasi-LU equivalence implies LU equivalence).

\section{Bipartite States}
In this section, we reformulate the statements and results in \cite{jing2016local} in terms of hypermatrix algebra. Note that the Fano form of a density matrix $\rho$ of a bipartite state of the quantum system $\mathbb{C}^{d_1}\otimes \mathbb{C}^{d_2}$ is given by 
\[\frac{1}{d_1d_2}\left(I_{d_1}\otimes I_{d_2}+\sum\limits_{i=1}^{\delta_1}T_1^i\lambda_i^1\otimes I_{d_2}+\sum\limits_{i=1}^{\delta_2}T_2^iI_{d_1}\otimes \lambda_i^2+\sum\limits_{i,j=1}^{\delta_1,\delta_2}T_{12}^{ij}\lambda_i^1\otimes \lambda_j^2\right),\]
and so its matrix representation is given by $\{T_1,T_2,T_{12}\}$. We say that the matrix representations $\{T_1,T_2,T_{12}\}$ and $\{\widehat{T}_1,\widehat{T}_2,\widehat{T}_{12}\}$ of $\rho$ and $\widehat{\rho}$ (respectively) are \textbf{simultaneously orthogonal (SO) equivalent} if there exists orthogonal matrices $O_1\in O(\delta_1),O_2\in O(\delta_2)$ such that 
\begin{equation}\label{SO equivalence - bipartite}
    \widehat{T}_{12} = (O_1,O_2)*T_{12}\quad \text{and}\quad \widehat{T}_1\circ \widehat{T}_2 = (O_1,O_2)*(T_1\circ T_2).
\end{equation}

In \cite{jing2016local}, it is shown that for bipartite states, quasi-LU equivalence is equivalent to SO equivalence (plus a norm assumption). However, in the special case of density matrices of $2$-qubit density matrices, it is not assumed that $O_1,O_2\in SO(3)$, only that they are in $O(3)$, in which case we cannot guarantee LU equivalence from quasi-LU equivalence. In the following, we present a modified proof of their result, showing that, in fact, we can guarantee LU equivalence from quasi-LU equivalence in the case of density matrices of $2$-qubits. We just need to add two more LU invariants to our assumption to obtain this guarantee! 

\begin{theorem}[quasi-LU Equivalence and SO Equivalence: Bipartite States]
	Let $\rho$ and $\widehat{\rho}$ be density matrices of bipartite states of the quantum system $\mathbb{C}^{d_1}\otimes \mathbb{C}^{d_2}$ with matrix representations $\{T_1,T_2,T_{12}\}$ and $\{\widehat{T}_1,\widehat{T}_2,\widehat{T}_{12}\}$ (respectively), and assume that $T_1,T_2,T_{12}\neq 0$. Then $\rho$ and $\widehat{\rho}$ are quasi-LU equivalent if and only if they are SO equivalent, and $\|\widehat{T}_i\| = \|T_i\|$ for some $i\in \{1,2\}$. 
    
    Furthermore, in the case where $\rho$ and $\widehat{\rho}$ are $2$-qubit density matrices, either conditions plus the assumption that 
    \[\mathrm{Det}(\widehat{T}_i\circ \widehat{T}_{12}) = \mathrm{Det}(T_i\circ T_{12})\neq 0\qquad (i=1,2)\]
    are equivalent to LU equivalence. 
\end{theorem}
\begin{proof}
	$\rho$ and $\widehat{\rho}$ are quasi-LU equivalent if and only if there exists orthogonal matrices $O_1\in O(\delta_1),O_2\in O(\delta_2)$ such that 
    \begin{align*}
        \widehat{T}_1 &= O_1*T_1 = O_1T_1, \\
        \widehat{T}_2 &= O_2*T_2 = O_2T_2, \quad\text{and} \\
        \widehat{T}_{12} &= (O_1,O_2)*T_{12} = O_1T_{12}O_2^t.
    \end{align*}
	Assuming quasi-LU equivalence, we immediately have $\widehat{T}_{12} = (O_1,O_2)*T_{12}$. Moreover, 
    \begin{align*}
        \widehat{T}_1\circ \widehat{T}_2 &= O_1T_1\circ O_2T_2 \\
        &= O_1T_1(O_2T_2)^t \\
        &= O_1(T_1T_2^t)O_2^t \\
        &= (O_1,O_2)*T_1\circ T_2,
    \end{align*}
	and
    \[\|\widehat{T}_i\| = \|O_i*T_i\| = \|O_iT_i\| = \|T_i\|\]
	for $i=1,2$. Thus, forward implication is proven.
	
	Conversely, suppose $\rho$ and $\widehat{\rho}$ are SO equivalent and without loss of generality assume that $\|\widehat{T}_1\|=\|T_1\|$. First, note that $T_1,T_2\neq 0$ implies that $\widehat{T}_1,\widehat{T}_2\neq 0$. Therefore, 
	\begin{align}
		0 &< \|\widehat{T}_1\|^2\|\widehat{T}_2\|^2 \label{T_1T_2 line1} \\
        &= \widehat{T}_1^t\widehat{T}_1\widehat{T}_2^t\widehat{T}_2 \label{T_1T_2 line2} \\
        &= \widehat{T}_1^t(\widehat{T}_1\circ \widehat{T}_2)\widehat{T}_2 \\
        &= \widehat{T}_1^t\big((O_1,O_2)*(T_1\circ T_2)\big)\widehat{T}_2,\quad \text{by assumption} \\
        &= \widehat{T}_1^t\big((O_1*T_1)\circ (O_2*T_2)\big)\widehat{T}_2,\quad \text{by equation \eqref{outerproduct - multilinmatmult - special case}} \\
        &= \widehat{T}_1^t\big((O_1T_1)\circ (O_2T_2)\big)\widehat{T}_2 \\
        &= \widehat{T}_1^t\big(O_1T_1(O_2T_2)^t\big)\widehat{T}_2 \\
        &= \widehat{T}_1^tO_1T_1T_2^tO_2^t\widehat{T}_2 \label{T_1T_2 line8}
	\end{align}
    Note in particular that this shows that 
    \begin{equation}\label{T_1T_2 outer product}
        \widehat{T}_1\circ \widehat{T}_2 = O_1T_1T_2^tO_2^t.
    \end{equation}
    Note also that $\widehat{T}_1^tO_1T_1$ and $T_2^tO_2^t\widehat{T}_2$ are just numbers, and so by lines \eqref{T_1T_2 line1}, \eqref{T_1T_2 line2}, and \eqref{T_1T_2 line8} we have that 
    \begin{equation}\label{alpha - bipartite case}
        \alpha:=\frac{\widehat{T}_1^tO_1T_1}{\widehat{T}_1^t\widehat{T}_1} = \frac{\widehat{T}_2^t\widehat{T}_2}{T_2^tO_2^t\widehat{T}_2} \neq 0.
    \end{equation}
	Therefore by equations \eqref{T_1T_2 outer product} and \eqref{alpha - bipartite case} we have that 
    \[\widehat{T}_1\widehat{T}_2\widehat{T}_2 = O_1T_1T_2O_2^t\widehat{T}_2\]
    which implies that 
    \[\widehat{T}_1\cdot\frac{\widehat{T}_2^t\widehat{T}_2}{T_2^tO_2^t\widehat{T}_2} = O_1T_1\]
    or in other words
    \begin{equation}\label{T_1 - bipartite case}
        \alpha\widehat{T}_1 = O_1T_1.
    \end{equation}
    Similarly, 
    \[\widehat{T}_1^t\widehat{T}_1\widehat{T}_2^t = \widehat{T}_1^tO_1T_1T_2^tO_2^t\]
    which implies that
    \[\frac{\widehat{T}_1^t\widehat{T}_1}{\widehat{T}_1^tO_1T_1}\cdot \widehat{T}_2^t = T_2^tO_2^t\]
    or in other words
    \[\alpha^{-1}\widehat{T}_2^t = T_2^tO_2^t\]
    and hence
    \begin{equation}\label{T_2 - bipartite case}
        \alpha^{-1}\widehat{T}_2 = O_2T_2.
    \end{equation}
    Now, by assumption $\|\widehat{T}_i\|=\|T_i\|$ for at least one $i\in \{1,2\}$, and in either case equation \eqref{T_1 - bipartite case} or equation \eqref{T_2 - bipartite case} imply that $|\alpha|=1$. Thus, $\alpha_i=\pm 1$ since $T_i$ and $T_{12}$ are real. If $\alpha=1$, then we have that
    \begin{align*}
        \widehat{T}_1 &= O_1T_1 \\
        \widehat{T}_2 &= O_2T_2 \\
        \widehat{T}_{12} &= O_1T_{12}O_2^t;
    \end{align*}
    on the other hand if $\alpha=-1$, then we have that
    \begin{align*}
        \widehat{T}_1 &= \overline{O}_1T_1 \\
        \widehat{T}_2 &= \overline{O}_2T_2 \\
        \widehat{T}_{12} &= \overline{O}_1T_{12}\overline{O}_2^t
    \end{align*}
    where $\overline{O}_i = -O_i\in O(\delta_i)$. In either case, quasi-LU equivalence is established. 

	Having established a correspondence between quasi-LU equivalence and SO equivalence, we know look at the special case where $\rho$ and $\widehat{\rho}$ are $2$-qubit density matrices. In particular, if we additionally assume that 
    \begin{equation}\label{Det - bipartite case}
        \mathrm{Det}(\widehat{T}_i\circ \widehat{T}_{12}) = \mathrm{Det}(T_i\circ T_{12})\qquad (i=1,2)
    \end{equation}
    then by equation \eqref{outerproduct - multilinmatmult - relation} and Proposition 1 it follows that equation \eqref{Det - bipartite case} reduces to
    \[\det(O_i)^2\det(O_{\overline{i}})\mathrm{Det}(T_i\circ T_{12}) = \mathrm{Det}(T_i,T_{12})\qquad (i\in \{1,2\})\]
    with $\overline{i} \in \{1,2\}\setminus \{i\}$, implying that $\det(O_{\overline{i}})$ must be positive. Thus, both $O_1,O_2\in SO(3)$, in which case quasi-LU equivalence implies LU equivalence, proving that $\rho$ and $\widehat{\rho}$ are LU equivalent.
\end{proof}
\subsection{SO Equivalence and Trace Identities: Bipartite States}
In \cite[Theorem 3.4]{jing2015unitary}, Jing generalizes Specht's criterion \cite{specht1940theorie,pearcy1962complete}, showing that for any two sets $\{A_i\}$ and $\{B_i\}$ of real $m\times n$ matrices, the following are equivalent:
\begin{itemize}
	\item $\{A_i\}$ is simultaneously orthogonal equivalent to $\{B_i\}$ (i.e. for each $i$, $B_i=OA_iP^t$ for some orthogonal matrices $O$ and $P$);
	\item $\{A_iA_i^t:i\leq j\}$ is simultaneously orthogonal similar to $\{B_iB_j^t:i\leq j\}$ (i.e. for each $i$ and $j$, $B_iB_j = OA_iA_jO^t$ for some orthogonal matrix $O$);
	\item $\Tr(w\{A_iA_j^t\}) = \Tr(w\{B_iB_j^t\})$ for any word $w$ in respective alphabets.
\end{itemize}
In general, it is difficult to show that two states are quasi-LU equivalent because we do not a priori know anything about the orthogonal matrices that relate the two states. However, Jing's Theorem 3.4 and our Theorem 1 imply the following:
\begin{theorem}[Characterization of quasi-LU Equivalence: Bipartite States] 
	Suppose $\rho$ and $\widehat{\rho}$ are two density matrices of bipartite states of the quantum system $\mathbb{C}^{d_1}\otimes \mathbb{C}^{d_2}$ with respective matrix representations $\{T_1,T_2,T_{12}\}$ and $\{\widehat{T}_1,\widehat{T}_2,\widehat{T}_{12}\}$, and assume that $T_1,T_2,T_{12}\neq 0$. Let $\{A_1,A_2\} = \{T_1\circ T_2,T_{12}\}$ and $\{B_1,B_2\} = \{\widehat{T}_1\circ \widehat{T}_2,\widehat{T}_{12}\}$. Then $\rho$ and $\widehat{\rho}$ are quasi-LU equivalent if and only if $\|T_i\| = \|\widehat{T}_i\|$ for at least one $i$ and 
    \begin{equation}\label{trace identities - bipartite case}
        \Tr(w\{A_iA_j^t\}) = \Tr(w\{B_iB_j^t\})\qquad (1\leq i\leq j\leq 2)
    \end{equation}
	for any word $w$. 
    
    Furthermore, in the case where $\rho$ and $\widehat{\rho}$ are density matrices of $2$-qubits, both conditions plus the assumption that 
    \[\mathrm{Det}(\widehat{T}_i\circ \widehat{T}_{12}) = \mathrm{Det}(T_i,T_{12})\qquad (i=1,2)\]
    are equivalent to LU equivalence. 
\end{theorem}
Thus, the problem of quasi-LU equivalence, or LU equivalence in the case where $\rho$ and $\widehat{\rho}$ are $2$-qubit density matrices, reduces to checking trace identities and a few other LU invariants. This is nice because in addition to being easy to compute, each trace identity is also an LU invariant, and so if any do not hold then we can conclude that $\rho$ and $\widehat{\rho}$ are not LU/quasi-LU equivalent; moreover, it is sufficient to check the trace identities for words of length at most $16(\delta_1+\delta_2)^2$ (see appendix or \cite{futorny2017specht}).  

\section{Tripartite States}
In this section, we will generalize Theorems 1 and 2 to tripartite quantum states of arbitrary dimension. This is done by introducing a hypermatrix algebra framework for quasi-LU and SO equivalence similar to how we presented them in the previous section, and then by applying a broader generalization of Specht's criterion \cite[Corollary 3.]{futorny2017specht}. 

First, note that the Fano form of a density matrix $\rho$ of a tripartite state of the quantum system $\mathbb{C}^{d_1}\otimes \mathbb{C}^{d_2}\otimes \mathbb{C}^{d_3}$ is given by
\begin{align*}
    &\quad\frac{1}{d_1d_2d_3}\bigg(I_{d_1}\otimes I_{d_2}\otimes I_{d_3} + \sum\limits_{i=1}^{\delta_1}T_1^i\lambda_i^1\otimes I_{d_2}\otimes I_{d_3}+\sum\limits_{i=1}^{\delta_2}T_2^iI_{d_1}\otimes \lambda_i^2\otimes I_{d_3} + \\
    &\qquad \sum\limits_{i=1}^{\delta_3}T_3^iI_{d_1}\otimes I_{d_2}\otimes \lambda_i^3 + \sum\limits_{i,j=1}^{\delta_1,\delta_2}T_{12}^{ij}\lambda_i^1\otimes \lambda_j^2\otimes I_{d_3} + \sum\limits_{i,j=1}^{\delta_1,\delta_3}T_{13}^{ij}\lambda_i^1\otimes I_{d_2}\otimes \lambda_j^3 + \\
    &\qquad\qquad \sum\limits_{i,j=1}^{\delta_2,\delta_3}T_{23}^{ij}I_{d_1}\otimes \lambda_i^2\otimes \lambda_j^3 + \sum\limits_{i,j,k=1}^{\delta_1,\delta_2,\delta_3}T_{123}^{ijk}\lambda_i^1\otimes \lambda_j^2\otimes \lambda_k^3\bigg)
\end{align*}
and so its hypermatrix representation is given by $\{T_1,T_2,T_3,T_{12},T_{13},T_{23},T_{123}\}$. For convenience set $T:=T_{123}$ and denote the $(i,j,k)$ entry of $T$ as $t_{ijk}$. Then by unfolding $T$, we may view it as any one of the following three matrices
\[T_{(1)} := \left[\begin{array}{ccccccccccccc}
	t_{111} & t_{121} & ... & t_{1\delta_21} & t_{112} & t_{122} & ... & t_{1\delta_22} & t_{113} & t_{123 }& ... & ... & t_{1\delta_2\delta_3} \\
	t_{211} & t_{221} & ... & t_{1\delta_21} & t_{212} & t_{222} & ... & t_{2\delta_22} & t_{213} & t_{223} & ... & ... & t_{2\delta_2\delta_3} \\
	\vdots & \vdots & \ddots & \vdots & \vdots & \vdots & \ddots & \vdots & \vdots & \vdots & \ddots & \ddots & \vdots \\
	t_{\delta_111} & t_{\delta_121} & ... & t_{\delta_1\delta_11} & t_{\delta_112} & t_{\delta_122} & ... & t_{\delta_1\delta_22} & t_{\delta_113} & t_{\delta_123} & ... & ... & t_{\delta_1\delta_2\delta_3}
\end{array}\right]_{\delta_1\times \delta_2\delta_3},\]
\[T_{(2)} := \left[\begin{array}{ccccccccccccc}
	t_{111} & t_{211} & ... & t_{\delta_111} & t_{112} & t_{212} & ... & t_{\delta_112} & t_{113} & t_{213} & ... & ... & t_{\delta_11\delta_3} \\
	t_{121} & t_{221} & ... & t_{\delta_121} & t_{122} & t_{222} & ... & t_{\delta_122} & t_{123} & t_{223} & ... & ... & t_{\delta_12\delta_3} \\
	\vdots & \vdots & \ddots & \vdots & \vdots & \vdots & \ddots & \vdots & \vdots & \vdots & \ddots & \ddots & \vdots \\
	t_{1\delta_21} & t_{2\delta_21} & ... & t_{\delta_1\delta_21} & t_{1\delta_22} & t_{2\delta_22} & ... & t_{\delta_1\delta_22} & t_{1\delta_23} & t_{2\delta_23} & ... & ... & t_{\delta_1\delta_2\delta_3} \\
\end{array}\right]_{\delta_2\times \delta_1\delta_3},\]
\[T_{(3)} := \left[\begin{array}{ccccccccccccc}
	t_{111} & t_{211} & ... & t_{\delta_111} & t_{121} & t_{221} & ... & t_{\delta_121} & t_{131} & t_{231} & ... & ... & t_{\delta_1\delta_21} \\
	t_{112} & t_{212} & ... & t_{\delta_112} & t_{122} & t_{222} & ... & t_{\delta_122} & t_{132} & t_{232} & ... & ... & t_{\delta_1\delta_22} \\
	\vdots & \vdots & \ddots & \vdots & \vdots & \vdots & \ddots & \vdots & \vdots & \vdots & \ddots & \ddots & \vdots \\
	t_{11\delta_3} & t_{21\delta_3} & ... & t_{\delta_11\delta_3} & t_{12\delta_3} & t_{22\delta_3} & ... & t_{\delta_12\delta_3} & t_{13\delta_3} & t_{23\delta_3} & ... & ... & t_{\delta_1\delta_2\delta_3}
\end{array}\right]_{\delta_3\times \delta_1\delta_2}.\]
Given $\delta_i\times \delta_i$ matrices $X_i$ ($1\leq i\leq 3$), direct computation yields
\begin{equation}\label{matricization - formula1}
    ((X_1,X_2,X_3)*T)_{(1)} = X_1T_{(1)}(X_3\otimes X_2)^t
\end{equation}
\begin{equation}\label{matricization - formula2}
    ((X_1,X_2,X_3)*T)_{(2)} = X_2T_{(2)}(X_3\otimes X_1)^t
\end{equation}
\begin{equation}\label{matricization - formula3}
    ((X_1,X_2,X_3)*T)_{(3)} = X_3T_{(3)}(X_2\otimes X_1)^t.
\end{equation}
Furthermore, if $v\in \mathbb{R}^{\delta_1}$ and $M\in \mathbb{R}^{\delta_2\delta_3}$, then direct computation yields
\begin{equation}\label{vectorization - formula1}
    (v\circ M)_{(1)} = v\circ \mathrm{vec}(M), 
\end{equation}
\begin{equation}\label{vectorization - formula2}
    (M\circ v)_{(3)}^t = \mathrm{vec}(M)\circ v,
\end{equation}
and 
\begin{equation}\label{vectorization - formula3}
    \mathrm{vec}((X_1,X_2)*M) = (X_2\otimes X_1)\mathrm{vec}(M).
\end{equation}

Now, if $\rho$ and $\widehat{\rho}$ are density matrices of tripartite states of the quantum system $\mathbb{C}^{d_1}\otimes \mathbb{C}^{d_2}\otimes \mathbb{C}^{d_3}$ with hypermatrix representations $\{T_1,T_2,T_3,T_{12},T_{13},T_{23},T_{123}\}$ and $\{\widehat{T}_1,\widehat{T}_2,\widehat{T}_3,\widehat{T}_{12},\widehat{T}_{13},\widehat{T}_{23},\widehat{T}_{123}\}$ (respectively), then we say they are \textbf{simultaneously orthogonal (SO) equivalent} if there exists $O_i\in O(\delta_i)$ such that 
\begin{equation}\label{SO Equivalence - Tripartite}
    \begin{split}
	\widehat{T}_{123} &= (O_1,O_2,O_3)*T_{123} \\
    \widehat{T}_1\circ \widehat{T}_{23} &= (O_1,O_2,O_3)*(T_1\circ T_{23}) \\
	\widehat{T}_2\circ \widehat{T}_{13} &= (O_2,O_1,O_3)*(T_2\circ T_{13}) \\
    \widehat{T}_{12}\circ \widehat{T}_3 &= (O_1,O_2,O_3)*(T_{12}\circ T_3).
    \end{split}
\end{equation}
As in the case of Theorem 1, we will establish that quasi-LU equivalence and SO equivalence are nearly equivalent notions. Then after applying a broader generalization of Specht's criterion \cite[Corollary 3.]{futorny2017specht}, we will show that quasi-LU equivalence, and hence LU equivalence in the case where $\rho$ and $\widehat{\rho}$ are $3$-qubit density matrices, essentially reduces to checking trace identities. 
\begin{theorem}[quasi-LU Equivalence and SO Equivalence: Tripartite States]
    Suppose $\rho$ and $\widehat{\rho}$ are two density matrices of tripartite states of the quantum system $\mathbb{C}^{d_1}\otimes \mathbb{C}^{d_2}\otimes \mathbb{C}^{d_3}$ with respective hypermatrix representations \\
    $\{T_1,T_2,T_3,T_{12},T_{13},T_{23},T_{123}\}$ and $\{\widehat{T}_1,\widehat{T}_2,\widehat{T}_3,\widehat{T}_{12},\widehat{T}_{13},\widehat{T}_{23},\widehat{T}_{123}\}$, and assume that $T_1,T_2,T_3$ and $T_{12},T_{13},T_{23}$ are nonzero. Then $\rho$ and $\widehat{\rho}$ are quasi-LU equivalent if and only if they are SO equivalent and the following conditions hold:
    \begin{enumerate}
        \item $\|\widehat{T}_i\| = \|T_i\|$ or $\|\widehat{T}_{jk}\| = \|T_{jk}\|$ for each $(i,j,k) = (1,2,3)$, $(2,1,3)$, and $(3,1,2)$; 
        \item $\widehat{T}_{i}^t\widehat{T}_{ij}\widehat{T}_{j}$ has the same sign as $T_{i}^tT_{ij}T_{j}$ for any one of $(i,j) = (1,2)$, $(2,3)$, or $(3,1)$, with $T_{ij} := T_{ji}^t$ if $i > j$. 
    \end{enumerate}
    Furthermore, in the case where $\rho$ and $\widehat{\rho}$ are $3$-qubit density matrices, either conditions plus the assumption that 
    \[\mathrm{Det}(\widehat{T}_i\circ \widehat{T}_{ij}) = \mathrm{Det}(T_i\circ T_{ij})\]
    for each $(i,j) = (1,2)$, $(1,3)$, and $(2,3)$, are equivalent to LU equivalence. 
\end{theorem}
\begin{proof}
	If $\rho$ and $\widehat{\rho}$ are quasi-LU equivalent, then there exists $O_i\in O(\delta_i)$ such that 
    \[\widehat{T}_i = O_i*T_i,\quad \widehat{T}_{jk} = (O_j,O_k)*T_{jk},\quad\text{and}\quad \widehat{T}_{123} = (O_1,O_2,O_3)*T_{123}\]
	for $1\leq i\leq 3$ and $1\leq j<k\leq 3$. It immediately follows that $\|\widehat{T}_i\|=\|T_i\|$ and $\|\widehat{T}_{jk}\|=\|T_{jk}\|$. Moreover by equation \eqref{outerproduct - multilinmatmult - relation}
    \[\widehat{T}_i\circ \widehat{T}_{jk} = (O_i*T_i)\circ ((O_j,O_k)*T_{jk}) = (O_i,O_j,O_k)*(T_i\circ T_{jk}),\]
	proving that $\rho$ and $\widehat{\rho}$ are SO-equivalent. 
	
	Conversely, analogous to as before, for $(i,j,k) = (1,2,3)$ and $(i,j,k) = (2,1,3)$ we have
	\begin{align}
		0 &< \|\widehat{T}_i\|^2\|\widehat{T}_{jk}\|^2 \label{T_iT_{jk}line0} \\
        &= (\widehat{T}_i^t\widehat{T}_i)(\text{vec}(\widehat{T}_{jk})^t\text{vec}(\widehat{T}_{jk})) \label{T_iT_{jk}line1} \\
		&= \widehat{T}_i^t(\widehat{T}_i\circ \text{vec}(\widehat{T}_{jk}))\text{vec}(\widehat{T}_{jk}) \\
		&= \widehat{T}_i^t((\widehat{T}_i\circ \widehat{T}_{jk})_{(1)})\text{vec}(\widehat{T}_{jk}) \text{\quad by equation }\eqref{vectorization - formula1} \\
		&= 	\widehat{T}_i^t\Big(\big((O_i,O_j,O_k)*(T_i\circ T_{jk})\big)_{(1)}\Big)\text{vec}(\widehat{T}_{jk}) \text{\quad by assumption} \label{T_iT_{jk}line2} \\
		&= \widehat{T}_i^t\Big(\big((O_i*T_i)\circ ((O_j,O_k)*T_{jk})\big)_{(1)}\Big)\text{vec}(\widehat{T}_{jk}) \text{\quad by equation }\eqref{outerproduct - multilinmatmult - relation} \\
		&= \widehat{T}_i^t\big((O_i*T_i)\circ \text{vec}((O_j,O_k)*T_{jk})\big)\text{vec}(\widehat{T}_{jk}) \text{\quad by equation }\eqref{vectorization - formula1} \\
        &= \widehat{T}_i^t\Big(O_iT_i\mathrm{vec}\big((O_j,O_k)*T_{jk}\big)^t\Big)\mathrm{vec}(\widehat{T}_{jk}) \\
		&= \widehat{T}_i^tO_iT_i\text{vec}(T_{jk})^t(O_k\otimes O_j)^t\text{vec}(\widehat{T}_{jk}) \text{\quad by equation }\eqref{vectorization - formula3} \label{T_iT_{jk}line3} 
	\end{align}
	Again, note that $\widehat{T}_i^tO_iT_i$ and $\text{vec}(T_{jk})^t(O_k\otimes O_j)^t\text{vec}(\widehat{T}_{jk})$ are just numbers, so lines \eqref{T_iT_{jk}line0}, \eqref{T_iT_{jk}line1} and \eqref{T_iT_{jk}line3} imply that 
    \begin{equation} \label{3quditalpha_i} 
        \alpha_i:=\frac{\widehat{T}_i^tO_iT_i}{\widehat{T}_i^t\widehat{T}_i} = \frac{\text{vec}(\widehat{T}_{jk})^t\text{vec}(\widehat{T}_{jk})}{\text{vec}(T_{jk})^t(O_k\otimes O_j)^t\text{vec}(\widehat{T}_{jk})}\neq 0.
    \end{equation} 
	Now by assumption, we also have that 
    \begin{equation} \label{T_iT_{jk}line4}
        (\widehat{T}_i\circ \widehat{T}_{jk})_{(1)} = \big((O_i,O_j,O_k)*(T_i\circ T_{jk})\big)_{(1)}
    \end{equation}
    The left-hand side of equation \eqref{T_iT_{jk}line4} can be expressed as $\widehat{T}_i\mathrm{vec}(\widehat{T}_{jk})^t$, and from equations \eqref{T_iT_{jk}line2} and \eqref{T_iT_{jk}line3} the right hand side is equal to $O_iT_i\mathrm{vec}(T_{jk})^t(O_k\otimes O_j)^t$. Therefore, 
    \[\widehat{T}_i\underbrace{\text{vec}(\widehat{T}_{jk})^t\text{vec}(\widehat{T}_{jk})}_{>0} = O_iT_i\text{vec}(T_{jk})^t(O_k\otimes O_j)^t\text{vec}(\widehat{T}_{jk})\]
	and so
    \[\widehat{T}_i = \frac{O_iT_i\text{vec}(T_{jk})^t(O_k\otimes O_j)^t\text{vec}(\widehat{T}_{jk})}{\text{vec}(\widehat{T}_{jk})^t\text{vec}(\widehat{T}_{jk})}.\]
	Thus, by equation \eqref{3quditalpha_i}
    \begin{eqnarray}\label{widehat-T_i - tripartite case}
        \begin{split}
            \alpha_i\widehat{T}_i &= \frac{\mathrm{vec}(\widehat{T}_{jk})^t\mathrm{vec}(\widehat{T}_{jk})}{\text{vec}(T_{jk})^t(O_k\otimes O_j)^t\mathrm{vec}(\widehat{T}_{jk})}\cdot \frac{O_iT_i\mathrm{vec}(T_{jk})^t(O_k\otimes O_j)^t\mathrm{vec}(\widehat{T}_{jk})}{\mathrm{vec}(\widehat{T}_{jk})^t\mathrm{vec}(\widehat{T}_{jk})} \\
            &= O_iT_i.
        \end{split}
    \end{eqnarray}
	Similarly, 
    \[\underbrace{\widehat{T}_i^t\widehat{T}_i}_{>0}\mathrm{vec}(\widehat{T}_{jk})^t = \widehat{T}_i^tO_iT_i\text{vec}(T_{jk})^t(O_k\otimes O_j)^t\]
	and so 
    \[\mathrm{vec}(\widehat{T}_{jk})^t = \frac{\widehat{T}_i^tO_iT_i\mathrm{vec}(T_{jk})^t(O_k\otimes O_j)^t}{\widehat{T}_i^t\widehat{T}_i} \Rightarrow \mathrm{vec}(\widehat{T}_{jk}) = \frac{(O_k\otimes O_j)\mathrm{vec}(T_{jk})T_i^tO_i^t\widehat{T}_i}{\widehat{T}_i^t\widehat{T}_i}.\]
	Thus, by equation \eqref{3quditalpha_i}
    \begin{align*}
        \alpha_i^{-1}\mathrm{vec}(\widehat{T}_{jk}) &= \frac{\widehat{T}_i^t\widehat{T}_i}{\widehat{T}_i^tO_iT_i}\cdot \frac{(O_k\otimes O_j)\mathrm{vec}(T_{jk})T_i^tO_i^t\widehat{T}_i}{\widehat{T}_i^t\widehat{T}_i} \\
        &= (O_k\otimes O_j)\mathrm{vec}(T_{jk}),
    \end{align*}
	or equivalently,
    \begin{equation}\label{widehat-T_iT_{jk} - tripartite case}
        \alpha_i^{-1}\widehat{T}_{jk} = (O_j,O_k)*T_{jk}    
    \end{equation}
    By assumption, either $\|\widehat{T}_i\|=\|T_i\|$ or $\|\widehat{T}_{jk}\|=\|T_{jk}\|$, and in either case it follows that $|\alpha_i|=1$, hence $\alpha_i=\pm 1$ since the $T_i$ and $T_{jk}$ are real. The calculations required to verify the case when $(i,j,k) = (3,1,2)$ are almost identical to that above (the only difference is that equation \eqref{vectorization - formula2} will be needed instead of \eqref{vectorization - formula1}), so they are omitted for the sake of brevity. Thus, for $(i,j,k) = (1,2,3), (2,1,3), (3,1,2)$ we have that
    \begin{eqnarray}\label{quasi-LU equivalence - Tripartite}
        \begin{split}
            \widehat{T}_i &= \alpha_i^{-1} O_i*T_i, \\
            \widehat{T}_{jk} &= \alpha_i(O_j,O_k)*T_{jk}, \\
            \widehat{T}_{123} &= (O_1,O_2,O_3)*T_{123},
        \end{split}
    \end{eqnarray}
	with $\alpha_i = \pm 1$. Now observe that for any one of $(i,j) = (1,2)$, $(2,3)$, or $(3,1)$ (with $T_{ij}:=T_{ji}^t$ if $i>j$), we have:
    \[\widehat{T}_i^t\widehat{T}_{ij}\widehat{T}_j = (\alpha_i^{-1}T_i^tO_i^t)(\alpha_kO_iT_{ij}O_j^t)(\alpha_j^{-1}O_jT_j) = \alpha_i\alpha_j\alpha_kT_i^tT_{ij}T_j,\]
    and so assuming that $\widehat{T}_i^t\widehat{T}_{ij}\widehat{T}_j$ has the same sign as $T_i^tT_{ij}T_j$, it follows that either each of the $\alpha$ values are positive or exactly two of the $\alpha$ values are negative. If they are all positive, then by equation \eqref{quasi-LU equivalence - Tripartite} quasi-LU equivalence is established. Suppose on the other hand that exactly two of the $\alpha$ values are negative; without loss of generality say $\alpha_1=1$ and $\alpha_2,\alpha_3=-1$. Then in this case by equation \eqref{quasi-LU equivalence - Tripartite} we have:
    \begin{align*}
        \widehat{T}_1 &= O_1T_1 \\
        \widehat{T}_2 &= -O_2T_2 = \overline{O}_2*T_2 \\
        \widehat{T}_3 &= -O_3T_3 = \overline{O}_3*T_3\\ 
        \widehat{T}_{12} &= -O_1T_{12}O_2^t = O_1T_{12}\overline{O}_2^t \\
        \widehat{T}_{13} &= -O_1T_{13}O_3^t = O_1T_{13}\overline{O}_3^t \\
        \widehat{T}_{23} &= O_2T_{23}O_3^t = \overline{O}_2T_{23}\overline{O}_3^t \\
        \widehat{T}_{123} &= (O_1,O_2,O_3)*T = (O_1,\overline{O}_2,\overline{O}_3)*T,
    \end{align*}
    which also establishes quasi-LU equivalence. 
 
    Lastly, in the case where $\rho$ and $\widehat{\rho}$ are $3$-qubit density matrices, if we additionally assume that
    \begin{equation}\label{Det - tripartite case}
        \mathrm{Det}(\widehat{T}_i\circ \widehat{T}_{ij}) = \mathrm{Det}(T_i\circ T_{ij})
    \end{equation}
    for each $(i,j) = (1,2)$, $(1,3)$, and $(2,3)$, then by equation \eqref{outerproduct - multilinmatmult - relation} and Proposition 1 it follows that equation \eqref{Det - tripartite case} reduces to
	\[\det(O_i)^2\det(O_j)\mathrm{Det}(T_i\circ T_{ij}) = \mathrm{Det}(T_i\circ T_{ij}),\]
    implying that $\det(O_j)=1$ for each $j=1,2,3$. Thus, $O_1,O_2,O_3\in SO(3)$, in which case quasi-LU equivalence implies LU equivalence, proving that $\rho$ and $\widehat{\rho}$ are LU equivalent. 
\end{proof}
Now, if $T_{jk} = 0$ for any $(j,k) = (1,2)$, $(1,3)$, or $(2,3)$, then Theorem 3 does not apply as currently stated. However, in such a case, if we replace the assumption in \eqref{SO Equivalence - Tripartite} that
\[\widehat{T}_i\circ \widehat{T}_{jk} = (O_i,O_j,O_k)*(T_i\circ T_{jk})\quad\text{or}\quad \widehat{T}_{jk}\circ \widehat{T}_i = (O_j,O_k,O_i)*(T_{jk}\circ T_i)\]
with the assumption that
\[\widehat{T}_i\circ \widehat{T}_j\circ\widehat{T}_k = (O_i,O_j,O_k)*(T_i\circ T_j\circ T_k),\]
and assume that $\|\widehat{T}_i\| = \|T_i\|$ for each $i$, then the proof of Theorem 3 holds without issue after replacing $\widehat{T}_{jk}$ with $\widehat{T}_j\circ \widehat{T}_k$; that is, the same exact logic applies and we obtain the same conclusion. Thus, we can relax the assumption that $T_{12},T_{13},T_{23}\neq 0$, as long as we modify the definition of SO equivalence and the norm assumption appropriately. 

\subsection{SO-Equivalence and Trace Identities: Tripartite States}
We would now like to apply Theorem 3.4 from \cite{jing2015unitary} as we did in the bipartite case, reducing the quasi-LU equivalence to verifying trace identities; however, there are some issues. In particular, suppose $\rho$ and $\widehat{\rho}$ are tripartite states of the quantum system $\mathbb{C}^{d_1}\otimes \mathbb{C}^{d_2}\otimes \mathbb{C}^{d_3}$ with hypermatrix representations $\{T_1,T_2,T_3,T_{12},T_{13},T_{23},T_{123}\}$ and $\{\widehat{T}_1,\widehat{T}_2,\widehat{T}_3,\widehat{T}_{12},\widehat{T}_{13},\widehat{T}_{23},\widehat{T}_{123}\}$, respectively. Define
\[\{A_1,A_2,A_3,A_4\} := \{(T_{123})_{(1)},(T_1\circ T_{23})_{(1)},(T_2\circ T_{13})_{(2)},(T_{12}\circ T_3)_{(1)}\}\] 
and similarly
\[\{B_1,B_2,B_3,B_4\} := \{(\widehat{T}_{123})_{(1)},(\widehat{T}_1\circ \widehat{T}_{23})_{(1)},(\widehat{T}_2\circ \widehat{T}_{13})_{(2)},(\widehat{T}_{12}\circ \widehat{T}_3)_{(1)}\}\]
(note that these are all $\delta_1\times \delta_2\delta_3$ matrices). Then by Theorem 3.4 in \cite{jing2015unitary}, there exists $O_1\in O(\delta_1)$ and $\widetilde{O}\in O(\delta_2\delta_3)$ such that 
\begin{align*}
	(\widehat{T}_{123})_{(1)} &= O_1(T_{123})_{(1)}\widetilde{O}^t \\
    (\widehat{T}_1\circ \widehat{T}_{23})_{(1)} &= O_1(T_1\circ T_{23})_{(1)}\widetilde{O}^t \\
	(\widehat{T}_2\circ \widehat{T}_{13})_{(2)} &= O_1(T_2\circ T_{13})_{(2)}\widetilde{O}^t \\
    (\widehat{T}_{12}\circ \widehat{T}_3)_{(1)} &= O_1(T_{12}\circ T_3)_{(1)}\widetilde{O}^t
\end{align*}
if and only if 
\[\Tr(w\{A_iA_j^t\}) = \Tr(w\{B_iB_j^t\})\qquad (1\leq i\leq j\leq 4)\]
for any word $w$. But this is not good enough! 

Indeed, we need $\widetilde{O}^t$ to be a tensor product of a $\delta_2\times \delta_2$ orthogonal matrix with a $\delta_3\times \delta_3$ orthogonal matrix; otherwise, we do not have SO equivalence and hence neither quasi-LU equivalence. To guarantee that $\widetilde{O}^t\in O(\delta_2)\otimes O(\delta_3)$, we need some additional assumptions and a more powerful generalization of Specht's criterion. 

\subsubsection{Applying the Partial Trace and Futorny's Theorem}
Let $\rho$ be a tripartite state of the quantum system $\mathbb{C}^{d_1}\otimes \mathbb{C}^{d_2}\otimes \mathbb{C}^{d_3}$ with hypermatrix representation \\
$\{T_1,T_2,T_3,T_{12},T_{13},T_{23},T_{123}\}$. Recall that the partial trace over the first subsystem $\Tr_1$ is defined as 
\[\Tr_1(\rho) := \sum\limits_{i=0}^{d_1-1}(\bra{i}\otimes I_{d_2}\otimes I_{d_3}\rho(\ket{i}\otimes I_{d_2}\otimes I_{d_3})\]
with $\{\ket{i}\}$ the computational basis for $\mathbb{C}^{d_1}$. This reduces to the density matrix of the following bipartite state
\[\frac{1}{d_2d_3}\left(I_{d_2}\otimes I_{d_3}+\sum\limits_{i=1}^{d_2^2-1}T_1^i\lambda_i^2\otimes I_{d_3}+\sum\limits_{i=1}^{d_3^2-1}T_2^iI_{d_2}\otimes \lambda_i^3+\sum\limits_{i,j=1}^{d_2^2-1,d_3^2-1}T_{12}^{ij}\lambda_i^2\otimes \lambda_j^3\right) =: \rho_{23}.\]
The partial traces over the second and third subsystems, $\Tr_2$ and $\Tr_3$ are analogously defined. Note that it is immediately clear that if two tripartite states $\rho$ and $\widehat{\rho}$ are quasi-LU equivalent, then their partial traces must also be quasi-LU equivalent; that is, the partial trace is a quasi-LU invariant.  

In \cite[Corollary 1]{futorny2017specht}, Futorny et al. apply the theory of quiver representations to generalize Specht's criterion \cite{specht1940theorie,pearcy1962complete} and Theorem 3.4 in \cite{jing2015unitary}. Their result and the necessary notions from quiver representation theory to understand it are discussed in the appendix. For our purposes, note that if we apply Corollary 1 in \cite{futorny2017specht} to the quiver
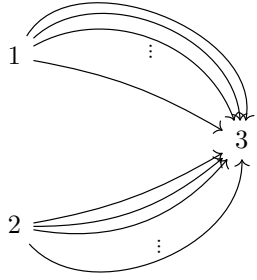
\begin{center}
    \begin{tikzcd}
        1 \arrow[bend left = 80, drrr] \arrow[bend left = 65,drrr] \arrow[bend left = 50, drrr,"\vdots",labels=below] \arrow[bend left=10,drrr] &&& \\
		&&& 3 \\
		2 \arrow[bend right = 10,urrr] \arrow[bend right = 20,urrr] \arrow[bend right = 30, urrr] \arrow[bend right = 70,urrr,"\vdots",labels=above] &&&  
    \end{tikzcd}
\end{center}
which has $k>0$ arrows from $1$ to $3$ and $l-k>0$ arrows from $2$ to $3$, and we assign to it a matrix representation $A$ with dimension $\dim(n_1,n_2,m)$, then we obtain the following result:
\begin{lemma}
    If $A_1,...,A_k,A_{k+1},...,A_l$ and $B_1,...,B_k,B_{k+1},...,B_l$ are real matrices each with $m$ rows, and with $A_1,...,A_k,B_1,...,B_k$ having $n_1$ columns and with $A_{k+1},...,A_l,B_{k+1},...,B_l$ having $n_2$ columns, then the following two statements are equivalent
    \begin{itemize}
        \item there exists orthogonal matrices $O\in O(m)$, $\widetilde{O}_1\in O(n_1)$, and $\widetilde{O}_2\in O(n_2)$ such that 
        \begin{equation}\label{lemma - condition1}
            (B_1,...,B_k,B_{k+1},...,B_l) = (OA_1\widetilde{O}_1,...,OA_k\widetilde{O}_1,OA_{k+1}\widetilde{O}_2,...,OA_l\widetilde{O}_2)
        \end{equation}
	   \item 
	   \begin{equation} \label{lemma - condition2}
        \begin{split} 
		  &\Tr(w\{A_{i_1}A_{j_1}^t,A_{i_2}A_{j_2}^t\big|1\leq i_1\leq j_1\leq k,k+1\leq i_2\leq j_2\leq l\}) \\
		  &= \Tr(w\{B_{i_1}B_{j_1}^t,B_{i_2}B_{j_2}^t\big|1\leq i_1\leq j_1\leq k,k+1\leq i_2\leq j_2\leq l\})
        \end{split}
	   \end{equation} 
	for every word $w$ in noncommuting variables. 
    \end{itemize}
    Furthermore, it suffices to check \eqref{lemma - condition2} for all words of length at most $[(r+2)(n_1+n_2+m)]^2$, where $r$ is the smallest positive integer such that $\frac{r(r+1)}{2}\geq \max\{k,l-k\}$.
\end{lemma}
Therefore, if for any one of $(i,j_1,j_2,k) = (1,2,2,3)$, $(2,1,1,3)$, or $(3,3,1,2)$
we set
\[(A_1,A_2,A_3,A_4,A_5) := ((T_{123})_{(i)},(T_1\circ T_{23})_{(i)},(T_2\circ T_{13})_{(j_1)},(T_{12}\circ T_3)_{(i)},T_i)\]
and 
\[(B_1,B_2,B_3,B_4,B_5) := ((\widehat{T}_{123})_{(i)},(\widehat{T}_1\circ \widehat{T}_{23})_{(i)},(\widehat{T}_2\circ \widehat{T}_{13})_{(j_1)},(\widehat{T}_{12}\circ \widehat{T}_3)_{(i)},\widehat{T}_i),\]
then by Lemma 1 there exists $O_i\in O(\delta_i)$, $\widetilde{O}\in O(\delta_{j_2}\delta_k)$, and $O\in O(1)$ (hence $O = [\pm 1]$, and without loss of generality we may assume that $O = [1]$) such that 
\begin{equation}\label{orthogonal similarity}
    \begin{split}
        (\widehat{T}_{123})_{(i)} &= O_i(T_{123})_{(i)}\widetilde{O}^t \\
        (\widehat{T}_1\circ \widehat{T}_{23})_{(i)} &= O_i(T_1\circ T_{23})_{(i)}\widetilde{O}^t \\
        (\widehat{T}_2\circ \widehat{T}_{13})_{(j_1)} &= O_i(T_2\circ T_{13})_{(j_1)}\widetilde{O}^t \\
        (\widehat{T}_{12}\circ \widehat{T}_3)_{(i)} &= O_i(T_{12}\circ T_3)_{(i)}\widetilde{O}^t \\
        \widehat{T}_i &= O_iT_i
    \end{split}
\end{equation}
if and only if 
\begin{equation} \label{initial trace identities}
    \Tr(w\{A_{\alpha}A_{\beta}^t,A_5A_5^t\}) = \Tr(w\{B_{\alpha}B_{\beta}^t,B_5B_5^t\})
\end{equation}
with $1\leq \alpha\leq \beta\leq 4$, for every word $w$ in noncommuting variables of length at most $25(1+\delta_i+\delta_{j_2}\delta_k)^2$. 

Suppose \eqref{initial trace identities} holds for all such words. If in addition we assume that $\Tr_i(\rho)$ is quasi-LU equivalent to $\Tr_i(\widehat{\rho})$ -which we can easily check by Theorem 2-, then there exists $O_{j_2}\in O(\delta_{j_2})$ and $O_k\in O(\delta_k)$ such that 
\begin{equation}\label{partial trace quasi-LU}
    \widehat{T}_{j_2} = O_{j_2}*T_{j_2},\qquad \widehat{T}_k = O_k*T_k,\qquad \text{and}\qquad \widehat{T}_{j_2k} = (O_{j_2},O_k)*T_{j_2k}.
\end{equation}
Consequently, from equations \eqref{outerproduct - multilinmatmult - relation} and \eqref{matricization - formula1}-\eqref{matricization - formula3} it follows that
\begin{equation}\label{case1}
    (T_i\circ T_{j_2k})_{(i)}^t(T_i\circ T_{j_2k})_{(i)}\widetilde{O}^t = (T_i\circ T_{j_2k})_{(i)}^t(T_i\circ T_{j_2k})_{(i)}(O_k\otimes O_{j_2})^t
\end{equation}
for $(i,j_1,j_2,k) = (1,2,2,3)$ or $(2,1,1,3)$, and 
\begin{equation}\label{case2}
    (T_{j_2k}\circ T_{i})_{(i)}^t(T_{j_2k}\circ T_{i})_{(i)}\widetilde{O}^t = (T_{j_2k}\circ T_{i})_{(i)}^t(T_{j_2k}\circ T_{i})_{(i)}(O_k\otimes O_{j_2})^t
\end{equation}
for $(i,j_1,j_2,k) = (3,3,1,2)$. 

Now, in the instances where $(i,j_1,j_2,k) = (1,2,2,3)$ or $(2,1,1,3)$,
if the square matrix $(T_i\circ T_{j_2k})_{(i)}^t(T_1\circ T_{23})_{(i)}$ of the equation \eqref{case1} is invertible, it follows that
\[\widetilde{O} = O_k\otimes O_{j_2};\]
similarly, in the case where $(i,j_1,j_2,k) = (3,3,1,2)$, if it happens that the square matrix $(T_{j_2k}\circ T_{i})_{(i)}^t(T_{j_2k}\circ T_{i})_{(i)}$ of \eqref{case2} is invertible, then it follows that 
\[\widetilde{O} = O_k\otimes O_{j_2}.\]
Therefore, in any such instance, we have
\begin{equation}
    \begin{split}
	   (\widehat{T}_{123})_{(i)} &= O_i(T_{123})_{(i)}(O_k\otimes O_{j_2})^t \\
       (\widehat{T}_1\circ \widehat{T}_{23})_{(i)} &= O_i(T_1\circ T_{23})_{(i)}(O_k\otimes O_{j_2})^t \\
	   (\widehat{T}_2\circ \widehat{T}_{13})_{(j_1)} &= O_i(T_2\circ T_{13})_{(j_1)}(O_k\otimes O_{j_2})^t \\ 
       (\widehat{T}_{12}\circ \widehat{T}_3)_{(i)} &= O_i(T_{12}\circ T_3)_{(i)}(O_k\otimes O_{j_2})^t.
    \end{split}
\end{equation}
This proves that $\rho$ and $\widehat{\rho}$ are SO equivalent, and hence quasi-LU equivalent -assuming also that the norm and sign conditions are met- by Theorem 3. Thus, we have the following theorem which characterizes quasi-LU equivalence:
\begin{theorem}[Characterizing quasi-LU Equivalence] 
	Suppose $\rho$ and $\widehat{\rho}$ are tripartite states of the quantum system $\mathbb{C}^{d_1}\otimes \mathbb{C}^{d_2}\otimes \mathbb{C}^{d_3}$ with respective matrix representations $\{T_1,T_2,T_3,T_{12},T_{13},T_{23},T_{123}\}$ and $\{\widehat{T}_1,\widehat{T}_2,\widehat{T}_2,\widehat{T}_{12},\widehat{T}_{13},\widehat{T}_{23},\widehat{T}_{123}\}$, and assume $T_1,T_2,T_3,T_{12},T_{13},T_{12}\neq 0$. For $(i,j_1,j_2,k) = (1,2,2,3)$, $(2,1,1,3)$, or $(3,3,1,2)$, let
    \[(A_1,A_2,A_3,A_4,A_5) := ((T_{123})_{(i)},(T_1\circ T_{23})_{(i)},(T_2\circ T_{13})_{(j_1)},(T_{12}\circ T_3)_{(i)},T_i)\]
    and 
    \[(B_1,B_2,B_3,B_4,B_5) := ((\widehat{T}_{123})_{(i)},(\widehat{T}_1\circ \widehat{T}_{23})_{(i)},(\widehat{T}_2\circ \widehat{T}_{13})_{(j_1)},(\widehat{T}_{12}\circ \widehat{T}_3)_{(i)},\widehat{T}_i).\]
    Then $\rho$ and $\widehat{\rho}$ are quasi-LU equivalent if and only if the following criteria are met:
    \begin{enumerate}
        \item $\|\widehat{T}_i\| = \|T_i\|$ or $\|\widehat{T}_{j_2k}\| = \|
        T_{j_2k}\|$ for each $(i,j,k) = (1,2,3)$, $(2,1,3)$, and $(3,1,2)$;
        \item $\widehat{T}_{i}^t\widehat{T}_{ij}\widehat{T}_{j}$ has the same sign as $T_{i}^tT_{ij}T_{j}$ for any one of $(i,j) = (1,2)$, $(2,3)$, or $(3,1)$, with $T_{ij} := T_{ji}^t$ if $i > j$;
        \item $\Tr(w\{A_{\alpha}A_{\beta}^t,A_5A_5^t\}) = \Tr(w\{B_{\alpha}B_{\beta}^t,B_5B_5^t\})$ with $1\leq \alpha\leq \beta\leq 5$, for all words $w$ in noncommuting variables of length at most $25(1+\delta_i+\delta_{j_2k})^2$ for at least one choice of $(i,j_1,j_2,k)$; 
        \item The partial traces $\Tr_i(\rho)$ and $\Tr_i(\widehat{\rho})$ are quasi-LU equivalent for the choice of $i$ in 3.; and lastly, 
        \item if $(i,j_1,j_2,k)$ is chosen to be $(1,2,2,3)$ or $(2,1,1,3)$ in part 3., then the matrix $(T_i\circ T_{j_2k})_{(i)}^t(T_i\circ T_{j_2k})_{(i)}$ is invertible; alternatively if $(i,j_1,j_2,k)$ is chosen to be $(3,3,1,2)$ in part 3., then the matrix $(T_{j_2k}\circ T_i)_{(i)}^t(T_{j_2k}\circ T_k)_{(i)}$ is invertible.
    \end{enumerate}
    Furthermore, in the case where $\rho$ and $\widehat{\rho}$ are $3$-qubit density matrices, either conditions plus the assumption that 
    \[\mathrm{Det}(\widehat{T}_i\circ \widehat{T}_{ij}) = \mathrm{Det}(T_i\circ T_{ij})\]
    for each $(i,j) = (1,2)$, $(1,3)$, and $(2,3)$, are equivalent to LU equivalence. 
\end{theorem}
By Theorem 2, checking the quasi-LU equivalence of the partial traces $\Tr_i(\rho)$ and $\Tr_i(\widehat{\rho})$ reduces to checking trace identities and norms. Thus, by Theorem 4, the difficult problem of determining quasi-LU equivalence between tripartite quantum states (or LU equivalence in the case of $3$-qubit density matrices) has \textit{essentially} been reduced to checking norms and traces (and the hyperdeterminant identities in the case of $3$-qubit density matrices), all of which are LU invariants and thus reasonable to assume and fairly easy to check. The only potential issues in the criteria stated in Theorem 4 are the second and last one: the sign assumption and the invertibility of the matrix $(T_i\circ T_{j_2k})_{(i)}^t(T_i\circ T_{j_2k})_{(i)}$ or $(T_{j_2k}\circ T_i)_{(i)}^t(T_{j_2k}\circ T_k)_{(i)}$ (depending on which choice of $(i,j_1,j_2,k)$ was chosen). Unfortunately, these are not LU invariants, and so consequently Theorem 3 may not apply to some tripartite states which are in fact LU/quasi-LU equivalent. Fortunately, since one has the freedom to choose which $(i,j_1,j_2,k)$ to pick (since any choice results in the same conclusion), there is some flexibility with criterion 5. Additionally, there is another way to at least partially overcome the unfortunate requirement of criterion 5.   

Indeed, if for the choice of $i$ made in criterion 3. we define $A_{i+1} := (T_i\circ T_{j_2}\circ T_3)_{(i)}$ and $B_{i+1} := (\widehat{T}_i\circ \widehat{T}_{j_2}\circ \widehat{T}_k)_{(i)}$, and if we replace the assumption that $\|\widehat{T}_{j_2k}\| = \|T_{j_2k}\|$ with $\|\widehat{T}_{j_2}\circ \widehat{T}_k\| = \|T_{j_2}\circ T_k\|$ in criteria 1. and 2., then assuming criteria 3. and 4. still hold, if $(T_i\circ T_{j_2}\circ T_k)_{(i)}^t(T_i\circ T_{j_2}\circ T_k)_{(i)}$ is invertible, then by applying the same logic as before we obtain
\[\widetilde{O} = (O_k\otimes O_{j_2}).\]
Therefore, by Theorem 3 and the subsequent discussion following the theorem, it follows that $\rho$ and $\widehat{\rho}$ are quasi-LU equivalent. Thus, even if $(T_i\circ T_{j_2k})_{(i)}^t(T_i\circ T_{j_2k})_{(i)}$ or $(T_{j_2k}\circ T_i)_{(i)}^t(T_{j_2k}\circ T_i)_{(i)}$ is not invertible, Theorem 4 may still hold if $(T_i\circ T_{j_2}\circ T_k)_{(i)}^t(T_i\circ T_{j_2k}\circ T_i)_{(i)}$ is; additionally, in such an instance, we may also relax the assumption that $T_{j_2k}\neq 0$. While this does not fully solve the issue of the invertibility criterion, it at least provides even greater flexibility, minimizing the burden caused by the invertibility criterion. 

\subsection{Extending to Arbitrary Multipartite States}
In section Section 2.2 we defined quasi-LU equivalence for density matrices of arbitrary multipartite mixed states. For a general hypermatrix $A\in F^{n_1\times...\times n_d}$ and matrices $X_1\in F^{\bullet\times n_1}$,..., $X_d\in F^{\bullet\times n_d}$, we have that 
\begin{equation}\label{generalized unfolding property}
    \big((X_1,...,X_d)*A\big)_{(k)} = X_kA_{(k)}(X_d\otimes X_{d-1}\otimes...\otimes X_{k+1}\otimes X_{k-1}\otimes...\otimes X_1)^t
\end{equation}
for $k=1,...,d$. Therefore, similarly we may consider SO equivalence for arbitrary multipartite quantum states, however there are some complications/technicalities. In particular, when defining SO equivalence for higher partite density matrices, we need to ensure that after unfolding the hypermatrices into matrices that they are being multiplied on the left and right by the same orthogonal matrices so that then we may be able to apply Futorney's Corollary 1. and reduce quasi-LU equivalence to trace identities. For instance, consider the $4$-partite case; in this instance, one might be tempted to define SO equivalence as the following:
\begin{align*}
    \widehat{T}_{1234} &= (O_1,O_2,O_3,O_4)*T \\
    \widehat{T}_i\circ \widehat{T}_{jkl} &= (O_i,O_j,O_k,O_l)*T_i\circ T_{jkl}\qquad 1\leq j<k<l\leq 4; i\in [4]\setminus \{j,k,l\} \\
    \widehat{T}_{ij}\circ \widehat{T}_{kl} &= (O_i,O_j,O_k,O_l)*(T_{ij}\circ T_{kl})\qquad 1\leq k<l\leq 4; i\neq j\in [4]\setminus \{k,l\}.
\end{align*}
However this will not work when applying the logic of Theorem 3 to try to establish a correspondence between quasi-LU equivalence and SO equivalence. The issue is that from equation \eqref{generalized unfolding property}, we see that the term $\widehat{T}_{1234}$ can be unfolded in only the following $4$ ways:
\begin{align*}
    (\widehat{T}_{1234})_{(1)} &= O_1(T_{1234})_{(1)}(O_4\otimes O_3\otimes O_2)^t \\
    (\widehat{T}_{1234})_{(2)} &= O_2(T_{1234})_{(2)}(O_4\otimes O_3\otimes O_1)^t \\
    (\widehat{T}_{1234})_{(3)} &= O_3(T_{1234})_{(3)}(O_4\otimes O_2\otimes O_1)^t \\
    (\widehat{T}_{1234})_{(4)} &= O_4(T_{1234})_{(4)}(O_3\otimes O_2\otimes O_1)^t,
\end{align*}
and the term $\widehat{T}_2\circ \widehat{T}_{134}$ can only be unfolded in only the following $4$ ways:
\begin{align*}
    (\widehat{T}_2\circ \widehat{T}_{134})_{(1)} &= O_2(T_2\circ T_{134})_{(1)}(O_4\otimes O_3\otimes O_1)^t \\
    (\widehat{T}_2\circ \widehat{T}_{134})_{(2)} &= O_1(T_2\circ T_{134})_{(2)}(O_4\otimes O_3\otimes O_2)^t \\
    (\widehat{T}_2\circ \widehat{T}_{134})_{(3)} &= O_3(T_2\circ T_{134})_{(3)}(O_4\otimes O_1\otimes O_2)^t \\
    (\widehat{T}_2\circ \widehat{T}_{134})_{(4)} &= O_4(T_2\circ T_{134})_{(4)}(O_3\otimes O_1\otimes O_2)^t,
\end{align*}
and lastly the term $\widehat{T}_3\circ \widehat{T}_{124}$ can be unfolded in the only the following $4$ ways:
\begin{align*}
    (\widehat{T}_3\circ \widehat{T}_{124})_{(1)} &= O_3(T_3\circ T_{124})_{(1)}(O_4\otimes O_2\otimes O_1)^t \\
    (\widehat{T}_3\circ \widehat{T}_{124})_{(2)} &= O_1(T_3\circ T_{124})_{(2)}(O_4\otimes O_2\otimes O_3)^t \\
    (\widehat{T}_3\circ \widehat{T}_{124})_{(3)} &= O_2(T_3\circ T_{124})_{(3)}(O_4\otimes O_1\otimes O_3)^t \\
    (\widehat{T}_3\circ \widehat{T}_{124})_{(4)} &= O_4(T_3\circ T_{124})_{(4)}(O_2\otimes O_1\otimes O_3)^t.
\end{align*}
So in particular, we cannot find a common right factor for each of the terms corresponding to the hypermatrices $\widehat{T}_{1234}$, $\widehat{T}_2\circ \widehat{T}_2\circ \widehat{T}_{134}$, and $\widehat{T}_3\circ \widehat{T}_{124}$. However, there is a way to overcome this problem. 

For any term of the form: $\widehat{T}_{i}\circ T_{jkl}$ or $T_{ij}\circ T_{kl}$, we may associate it to the corresponding permutation of the subscripts $\pi = ijkl$. If all such permutations $\pi$ are obtained either the permutation $1234$ or obtained from $1234$ by multiplying by a transposition, then it follows that there will be a common right factor in the matrix unfolding. Dobes and Jing proved in \cite{dobes2024qubits} that if the hypermatrix $B$ is obtained from $B$ by permuting the subscripts, then they are necessarily LU equivalent. Consequently, instead of considering the hypermatrix $\widehat{T}_3\circ \widehat{T}_{124}$, which has corresponding permutation $\pi = 3124$ and which is not obtained from the permutation $1234$ by multiplying by a transposition (in particular $\pi$ is obtained by the product $(1234)(132)$), we may instead consider the hypermatrix $\widehat{T}_3\circ \widehat{T}_{214}$ (in which case the resulting permutation $3214$ is obtained from the product $(1234)(23)$). Ensuring that for each term of the form: $\widehat{T}_{i}\circ T_{jkl}$ or $T_{ij}\circ T_{kl}$, their corresponding subscripts differ at most only by a transposition guarantees that there is a common left and right factor for each term after unfolding. From here, we may then apply very similar logic used in proving Theorems 3 and 4 to establish a connection between quasi-LU equivalence and SO equivalence, and then \textit{essentially} reduce both checking trace identities. 

There are of course, more technicalities to consider (hence why essentially was italicized). In higher partite cases we will also need assumptions similar to criteria 1., 2., and 4. in Theorem 4. Also, if we are working with density matrices of $n$-qubits, then to guarantee LU equivalence we will need to check similar hyperdeterminant properties, which may be hard to compute. All of this is hypothetically possibly, however admittedly a bit cumbersome in even just the $4$-partite case; furthermore, the number of trace identities grows rapidly as we increase each part. Thus, while we can conclude that the problem of determining whether or not two $n$-partite states $\rho$ and $\widehat{\rho}$ of arbitrary dimension/parts are quasi-LU equivalent (or LU equivalent in the instance where they are $n$-qubit density matrices) \textit{essentially} reduces to checking trace identities, our methods are difficult to carry out in higher partite cases beyond perhaps the $4$-partite case. In practice, our methods will be most useful in the bipartite and tripartite cases. Therefore, we present our work not as a successor to the work of Li and Qiao in \cite{li2013classification}, but rather as an alternative. Perhaps in the case of degenerate $3$, $4$, or possibly even $5$-partite states where their method breaks down, one might want to consider our approach and/or use our proofs as an outline for suitable generalizations. 

\section{Conclusion}
In summary, we extended the method and results of Jing et al. in \cite{jing2016local} to density matrices of tripartite quantum states. Additionally, from our hypermatrix algebra framework, it is apparent that the proofs and results extend to density matrices of more general multipartite quantum states. However, there are practical limitations to extending our methods to states of many parts since the computations required grow rapidly with each additional part, and additionally there are a number of technicalities to consider. Thus, we think our approach to characterizing LU equivalence is best to be considered in conjunction with the methods of \cite{li2013classification}, perhaps applied to special classes of states based on their Fano form. 

\appendix
\section{Appendix}
The following results in this appendix are valid for complex inner product spaces and complex Euclidean spaces. However, they will be stated in terms of real Euclidean spaces since that is what is needed for the purposes of this paper. 
\subsection{Specht's Criterion and Generalizations}
The original, complex version of Specht's criterion, proved by Wilhelm Specht, can be found in \cite{specht1940theorie}. The real version of Specht's criterion, which was proven by Carl Pearcy in \cite{pearcy1962complete} is stated below:
\begin{proposition}[Specht's criterion (real version)]
    If $A$ and $B$ are real $n\times n$ matrices, then they are orthogonally similar, i.e. $B = O^tAO$ for some $n\times n$ orthogonal matrix $O$, if and only if 
    \begin{equation}\label{specht - real version}
        \Tr(w\{A,A^t\}) = \Tr(w\{B,B^t\})
    \end{equation}
    for every word $w(x,y)$ in two noncommuting variables. 
\end{proposition}
One issue with Specht's criterion is that, as originally stated, it required infinitely many trace identities to check. However, in 1962 Carl Pearcy \cite{pearcy1962complete} that it suffices to verify conditions \eqref{specht - real version} for all words of length at most $2n^2$, and in 1986 T. Laffey showed in \cite{laffey1986simultaneous} that it suffices to verify conditions \eqref{specht - real version} for all words of length at most $\frac{2}{3}(n^2+2)$. Other tighter upper bounds have been derived, but note that $\frac{2}{3}(n^2+2)<n^2$ for all positive integers $n>2$, so for our purposes we conclude that (by Laffey's upper bound) it is enough to check \eqref{specht - real version} for all words of length at most $n^2$. 

In this paper, two generalizations of Specht's Criterion are considered, one due to Naihuan Jing \cite{jing2015unitary} and the other due to Futorny et al. \cite{futorny2017specht}. Jing's generalization is in fact a special case of Futorny's, so the rest of this appendix will be dedicated to reviewing Futorny's result. To better understand what Futorny et al. proved, we first need to review the basics of quiver representation theory. 

\subsubsection{Quiver Representations and Futorny's Theorem}
A \textbf{quiver} is a directed graph (loops and multiple arrows are allowed) used to represent vector spaces and algebras. A \textbf{representation} $\mathcal{A} = (\mathcal{A}_{\alpha},\mathcal{U}_v)$ of a quiver $Q$ over a field $F$ is given by assigning to each vertex $v$ a vector space $\mathcal{U}_v$ over $F$ to each arrow $\alpha:u\rightarrow v$ a linear transformation $\mathcal{A}_{\alpha}:u\rightarrow v$. The vector $\dim(\mathcal{A}) := (\dim(\mathcal{U}_1),...,\dim(\mathcal{U}_t))$ is called the \textbf{dimension} of the representation $\mathcal{A}$. 

An \textbf{oriented cycle} $\pi$ of length $l\geq 1$ in a quiver $Q$ is a sequence of arrows of the form
\begin{center}
	\begin{tikzcd}
		\pi: & v_1 \arrow[bend right, rrr, "\alpha_l",labels=above] & v_2 \arrow[l,"\alpha_1",labels=above] & ... \arrow[l,"\alpha_2",labels=above] & v_l \arrow[l,"\alpha_{l-1}",labels=above] 
	\end{tikzcd}
\end{center}
That is, it is a closed-directed walk (note that some vertices and arrows may repeat due to possible loops). For each representation $\mathcal{A}$ of a quiver $Q$ and any cycle $\pi$, we define $\mathcal{A}(\pi)$ to be the cycle of linear transformations
\begin{center}
	\begin{tikzcd}
		\mathcal{A}(\pi): & \mathcal{U}_{v_1} \arrow[bend right, rrr, "\mathcal{A}_{\alpha_l}",labels=above] & \mathcal{U}_{v_2} \arrow[l,"\mathcal{A}_{\alpha_1}",labels=above] & ... \arrow[l,"\mathcal{A}_{\alpha_2}",labels=above] & \mathcal{U}_{v_l} \arrow[l,"\mathcal{A}_{\alpha_{l-1}}",labels=above] 
	\end{tikzcd}
\end{center}
The \textbf{trace} of $\mathcal{A}(\pi)$ is defined as $\mathrm{trace}(\mathcal{A}(\pi)) := \mathrm{trace}(\mathcal{A}_{\alpha_1}\mathcal{A}_{\alpha_2}...\mathcal{A}_{\alpha_l})$. Note that $\mathrm{trace}(\mathcal{A}(\pi))$ does not depend on the choice of the initial vertex $v_1$ in the cycle since the trace is invariant under cyclic permutations. 

For each linear transformation $\mathcal{A}:\mathcal{U}\rightarrow \mathcal{V}$ between (real) Euclidean spaces $\mathcal{U}$ and $\mathcal{V}$, the \textbf{adjoint map} $\mathcal{A}^*:\mathcal{V}\rightarrow \mathcal{U}$ is given by $\langle \mathcal{A}x,y\rangle = \langle x,\mathcal{A}^*y\rangle$ for all $x\in \mathcal{U}$ and $y\in \mathcal{V}$. For a quiver $Q$ with vertices $v_1,...,v_t$, a \textbf{(real) Euclidean representation} $\mathcal{A} = (\mathcal{A}_{\alpha},\mathcal{U}_v)$ is given by assigning to each vertex $v$ a real Euclidean space $\mathcal{U}_v$, and to each arrow $\alpha:u\rightarrow v$ a linear transformation $\mathcal{A}_{\alpha}:\mathcal{U}_u\rightarrow \mathcal{U}_v$. Two (real) Euclidean representations $\mathcal{A} = (\mathcal{U}_{\alpha},\mathcal{U}_v)$ and $\mathcal{B} = (\mathcal{B}_{\alpha},\mathcal{V}_v)$ of $Q$ are \textbf{isometric} if there exists a family isometries (i.e. linear isomorphisms that preserve inner products) $\varphi_1:\mathcal{U}_1\rightarrow \mathcal{V}_1$,..., $\varphi_t:\mathcal{U}_t\rightarrow \mathcal{V}_t$  such that the diagram
\begin{center}
    \begin{tikzcd}
        \mathcal{U}_u \arrow[r,"\mathcal{A}_{\alpha}"] \arrow[d,"\varphi_u",labels=left] & \mathcal{U}_v \arrow[d,"\varphi_v"] \\
        \mathcal{V}_u \arrow[r,"\mathcal{B}_{\alpha}"] & \mathcal{V}_V
    \end{tikzcd}
\end{center}
commutes (i.e $\varphi_v\mathcal{A}_{\alpha} = \mathcal{B}_{\alpha}\varphi_u)$ for each arrow $\alpha:u\rightarrow v$. 

Now, for each quiver $Q$, we denote $\widetilde{Q}$ to be the quiver with double the number of arrows in $Q$, obtained from $Q$ by attaching the arrow $\alpha^*:v\rightarrow u$ for each arrow $\alpha:u\rightarrow v$ in $Q$. For each (real) Euclidean representation $\mathcal{A}$ of $Q$, we define the (real) Euclidean representation $\widetilde{A}$ of $\widetilde{Q}$ that coincides with $\mathcal{A}$ on $Q\subset \widetilde{Q}$ and that assigns to each new arrow $\alpha^*:v\rightarrow u$ the linear transformation $\widetilde{\mathcal{A}}_{\alpha^*}:=\mathcal{A}_{\alpha}^*:\mathcal{U}_v\rightarrow \mathcal{U}_u$ (i.e. the adjoint of $\mathcal{A}_{\alpha}$). For example, if $Q$ is given by 
\begin{center}
    \begin{tikzcd}
        Q: && u \arrow[loop left,"\beta"] \arrow[r,"\alpha"] & v 
    \end{tikzcd}
\end{center}
and $\mathcal{A}$ is a representation on $Q$ given by
\begin{center}
    \begin{tikzcd}
        \mathcal{A}: && \mathcal{U}_u \arrow[loop left,"\mathcal{A}_{\beta}"] \arrow[r,"\mathcal{A}_{\alpha}"] & \mathcal{U}_v
    \end{tikzcd}
\end{center}
then $\widetilde{Q}$ is given by 
\begin{center}
    \begin{tikzcd}
        \widetilde{Q}: && u \arrow[loop left,"\beta"] \arrow[loop below,"\beta^*"] \arrow[r,bend left,"\alpha"] & v \arrow[l,bend left,"\alpha^*"]
    \end{tikzcd}
\end{center}
then $\widetilde{\mathcal{A}}$ is given by 
\begin{center}
    \begin{tikzcd}
        \widetilde{\mathcal{A}}: && \mathcal{U}_u \arrow[loop left,"\mathcal{A}_{\beta}"] \arrow[loop below,"\widetilde{\mathcal{A}}_{\beta^*}"] \arrow[r,bend left,"\mathcal{A}_{\alpha}"] & \mathcal{U}_v \arrow[l,bend left,"\widetilde{\mathcal{A}}_{\alpha^*}"]
    \end{tikzcd}
\end{center}

Representations of quivers can be expressed in terms of matrices. If $[x]$ is the coordinate vector of $x\in \mathcal{U}$ in some orthonormal basis, then $\langle x,y\rangle = [x]^T [y]$ for all $x,y\in \mathcal{U}$. Furthermore, if $A$ is the matrix of the linear transformation $\mathcal{A}:\mathcal{U}\rightarrow \mathcal{V}$ in some orthonormal bases for $\mathcal{U}$ and $\mathcal{V}$, then $A^t$ is the matrix of the adjoint transformation $\mathcal{A}^*:\mathcal{V}\rightarrow \mathcal{U}$. A \textbf{matrix representation $A$ of dimension} $(d_1,...,d_t)$ of a quiver $Q$ is given by assigning to each arrow $\alpha:u\rightarrow v$ a matrix $A_{\alpha}$ of size $d_v\times d_u$ (note that we take $d_i:=0$ if the vertex $i$ does not have arrows). Two (real) Euclidean matrix representations $A$ and $B$ of $Q$ are \textbf{isometric} if there exists orthogonal matrices $O_1,...,O_t$ such that 
\[B_{\alpha} = O_v^{-1}A_{\alpha}O_u = O_v^tA_{\alpha}O_u\]
for every arrow $\alpha:u\rightarrow v$. Note that $\mathcal{A}$ and $\mathcal{B}$ are isometric if and only if $A$ and $B$ are isometric. Furthermore, if $\widetilde{\mathcal{A}}$ is the corresponding representation obtained from a representation $\mathcal{A}$ as described in the previous paragraph, then $\widetilde{A}$ is the matrix form of $\widetilde{\mathcal{A}}$. Lastly, for each oriented cycle $\pi$ in a quiver $Q$ and each matrix representation $A$ of $Q$, we denote $A(\pi) := A_{\alpha_1}A_{\alpha_2}...A_{\alpha_l}$. 

With this framework of quiver representations, Futorny et al. then prove in \cite{futorny2017specht} the following important theorem.
\begin{theorem}
    Two (real) Euclidean matrix representations $A$ and $B$ of a quiver $Q$ are isometric if and only if 
    \begin{equation}\label{futorney trace identities}
        \mathrm{trace}(\widetilde{A}(\pi)) = \mathrm{trace}(\widetilde{B}(\pi))
    \end{equation}
    for each oriented cycle $\pi$ in the quiver $\widetilde{Q}$. Moreover, it suffices to verify \eqref{futorney trace identities} for all cycles $\pi$ of length at most 
    \[\varphi((r+2)(d_1+...+d_t)),\]
    where $\varphi(n)$ is any bound for the sufficient word length in Specht's criterion (e.g. $\varphi(n)=n^2$) and $r$ is the minimal natural number such that 
    \[\frac{r(r+1)}{2}\geq \max\{m_{ij}|i\text{ and }j\text{ are vertices of }Q\}\]
    in which $m_{ij}$ is the number of arows from $j$ to $i$ in $Q$. 
\end{theorem}
Note that if we assign a matrix representation $A$ of dimension $n$ to the following quiver
\begin{center}
    \begin{tikzcd}
        Q: && 1 \arrow[loop left]
    \end{tikzcd}
\end{center}
and then apply Theorem 5, then we obtain Specht's Criterion. In \cite{jing2015unitary}, the Jing proves the following generalization of Specht's criterion:
\begin{proposition}
    Let $(A_1,...,A_k)$ and $(B_1,..,B_k)$ be two $k$-tuples of $m\times n$ matrices. Then there exists orthogonal matrices $O$ and $P$ such that 
    \begin{equation}
        (B_1,...,B_k) = (OA_1P,...,OA_kP)
    \end{equation}
    if and only if 
    \begin{equation}
        \Tr(w\{A_1^tA_1,...,A_i^tA_j,...,A_k^tA_k\}) = \Tr(w\{B_1^tB_1,...,B_i^tB_j,...,B_k^tB_k\})
    \end{equation}
    for every word $w(x_{11},...,x_{ij},...,x_{kk})$ in $k^2$ noncommuting variables. 
\end{proposition}
Indeed, by assigning two matrix representations $A$ and $B$ of dimension $(n,m)$ on the following quiver 
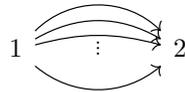
\begin{center}
    \begin{tikzcd}
		1 \arrow[bend left = 45, rr] \arrow[bend left = 30, rr] \arrow[bend left = 15, rr, "\vdots",labels=below] \arrow[bend right=40,rr] && 2 
    \end{tikzcd}
\end{center}
and then applying Theorem 5, we obtain Jing's generalization of Specht's criterion. Lastly, note that by attaching two pairs of matrix representations $A$ and $B$ of dimensions $(\delta_{j_2}\delta_k,1,\delta_i)$ to the quiver
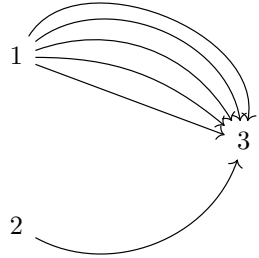
\begin{center}
    \begin{tikzcd}
        1 \arrow[bend left = 80, drrr] \arrow[bend left = 60, drrr] \arrow[bend left=40,drrr] \arrow[bend left=20,drrr] \arrow[drrr] &&& \\
		&&& 3 \\
		2 \arrow[bend right=50, urrr] &&&  
    \end{tikzcd}
\end{center}
and then applying Theorem 5, we obtain the equations \eqref{orthogonal similarity} from section 5.1.1.

\bibliographystyle{ieeetr}
\bibliography{references}

\begin{thebibliography}{10}

\bibitem{einstein1935can}
A.~Einstein, B.~Podolsky, and N.~Rosen, ``Can quantum-mechanical description of physical reality be considered complete?,'' {\em Physical review}, vol.~47, no.~10, p.~777, 1935.

\bibitem{nielsen2001quantum}
M.~A. Nielsen and I.~L. Chuang, {\em Quantum computation and quantum information}, vol.~2.
\newblock Cambridge university press Cambridge, 2001.

\bibitem{bengtsson2017geometry}
I.~Bengtsson and K.~{\.Z}yczkowski, {\em Geometry of quantum states: an introduction to quantum entanglement}.
\newblock Cambridge university press, 2017.

\bibitem{yu2021advancements}
Y.~Yu, ``Advancements in applications of quantum entanglement,'' in {\em Journal of Physics: Conference Series}, vol.~2012, p.~012113, IOP Publishing, 2021.

\bibitem{walter2016multipartite}
M.~Walter, D.~Gross, and J.~Eisert, ``Multipartite entanglement,'' {\em Quantum Information: From Foundations to Quantum Technology Applications}, pp.~293--330, 2016.

\bibitem{monras2011entanglement}
A.~Monras, G.~Adesso, S.~M. Giampaolo, G.~Gualdi, G.~B. Davies, and F.~Illuminati, ``Entanglement quantification by local unitary operations,'' {\em Physical Review A—Atomic, Molecular, and Optical Physics}, vol.~84, no.~1, p.~012301, 2011.

\bibitem{li2014local}
M.~Li, T.~Zhang, S.-M. Fei, X.~Li-Jost, and N.~Jing, ``Local unitary equivalence of multiqubit mixed quantum states,'' {\em Physical Review A}, vol.~89, no.~6, p.~062325, 2014.

\bibitem{makhlin2002nonlocal}
Y.~Makhlin, ``Nonlocal properties of two-qubit gates and mixed states, and the optimization of quantum computations,'' {\em Quantum Information Processing}, vol.~1, pp.~243--252, 2002.

\bibitem{kraus2010local2}
B.~Kraus, ``Local unitary equivalence of multipartite pure states,'' {\em Physical review letters}, vol.~104, no.~2, p.~020504, 2010.

\bibitem{liu2012local}
B.~Liu, J.-L. Li, X.~Li, and C.-F. Qiao, ``Local unitary classification of arbitrary dimensional multipartite pure states,'' {\em Physical review letters}, vol.~108, no.~5, p.~050501, 2012.

\bibitem{li2013classification}
J.-L. Li and C.-F. Qiao, ``Classification of arbitrary multipartite entangled states under local unitary equivalence,'' {\em Journal of Physics A: Mathematical and Theoretical}, vol.~46, no.~7, p.~075301, 2013.

\bibitem{jing2016local}
N.~Jing, M.~Yang, and H.~Zhao, ``Local unitary equivalence of quantum states and simultaneous orthogonal equivalence,'' {\em Journal of Mathematical Physics}, vol.~57, no.~6, 2016.

\bibitem{jing2015unitary}
N.~Jing, ``Unitary and orthogonal equivalence of sets of matrices,'' {\em Linear Algebra and its Applications}, vol.~481, pp.~235--242, 2015.

\bibitem{futorny2017specht}
V.~Futorny, R.~A. Horn, and V.~V. Sergeichuk, ``Specht's criterion for systems of linear mappings,'' {\em Linear Algebra and its Applications}, vol.~519, pp.~278--295, 2017.

\bibitem{lim2013tensors}
L.-H. Lim, ``Tensors and hypermatrices,'' {\em Handbook of linear algebra}, vol.~2, 2013.

\bibitem{kolda2009tensor}
T.~G. Kolda and B.~W. Bader, ``Tensor decompositions and applications,'' {\em SIAM review}, vol.~51, no.~3, pp.~455--500, 2009.

\bibitem{bertlmann2008bloch}
R.~A. Bertlmann and P.~Krammer, ``Bloch vectors for qudits,'' {\em Journal of Physics A: Mathematical and Theoretical}, vol.~41, no.~23, p.~235303, 2008.

\bibitem{de2011multipartite}
J.~I. de~Vicente and M.~Huber, ``Multipartite entanglement detection from correlation tensors,'' {\em Physical Review A—Atomic, Molecular, and Optical Physics}, vol.~84, no.~6, p.~062306, 2011.

\bibitem{specht1940theorie}
W.~Specht, ``Zur theorie der matrizen. ii.,'' {\em Jahresbericht der Deutschen Mathematiker-Vereinigung}, vol.~50, pp.~19--23, 1940.

\bibitem{pearcy1962complete}
C.~Pearcy, ``A complete set of unitary invariants for operators generating finite w\^{}*-algebras of type i.,'' 1962.

\bibitem{dobes2024qubits}
I.~Dobes and N.~Jing, ``Qubits as hypermatrices and entanglement,'' {\em Physica Scripta}, vol.~99, no.~5, p.~055110, 2024.

\bibitem{laffey1986simultaneous}
T.~J. Laffey, ``Simultaneous reduction of sets of matrices under similarity,'' {\em Linear Algebra and Its Applications}, vol.~84, pp.~123--138, 1986.

\end{thebibliography}

\end{document}